\documentclass[lettersize,journal]{IEEEtran}
\usepackage{amsmath,amsfonts,amssymb}
\usepackage[linesnumbered,ruled,vlined]{algorithm2e}
\SetKwComment{Comment}{/*}{}
\usepackage{array}
\usepackage[caption=false,font=normalsize,labelfont=sf,textfont=sf]{subfig}
\usepackage{textcomp}
\usepackage{stfloats}
\usepackage{url}
\usepackage{verbatim}
\usepackage{graphicx}
\usepackage{cite}
\def\BibTeX{{\rm B\kern-.05em{\sc i\kern-.025em b}\kern-.08em
    T\kern-.1667em\lower.7ex\hbox{E}\kern-.125emX}}
\usepackage{balance}
\usepackage{multirow}
\usepackage{booktabs}
\usepackage{threeparttable}
\usepackage{arydshln}

\usepackage{makecell}
\usepackage[english]{babel}
\usepackage{amsthm}
\usepackage{diagbox}
\usepackage{xcolor}

{ 
\theoremstyle{definition}
\newtheorem{definition}{Definition}
\newtheorem{assumption}{Assumption}
\newtheorem{proposition}{Proposition}

\newtheorem{remark}{Remark}
}
\begin{document}

\title{Dimension-reduced Optimization of Multi-zone Thermostatically Controlled Loads}

\author{Xueyuan Cui, \IEEEmembership{Student Member, IEEE,} Yi Wang, \IEEEmembership{Senior Member, IEEE,} and Bolun Xu, \IEEEmembership{Member, IEEE}
}

\markboth{Submitted to IEEE Trans. Smart Grid}%
{Shell \MakeLowercase{\textit{et al.}}: A Sample Article Using IEEEtran.cls for IEEE Journals}

\maketitle

\begin{abstract}
This study proposes a computationally efficient method for optimizing multi-zone thermostatically controlled loads (TCLs) by leveraging dimensionality reduction through an auto-encoder. We develop a multi-task learning framework to jointly represent latent variables and formulate a state-space model based on observed TCL operation data. This significantly reduces the dimensionality of TCL variables and states while preserving critical nonlinear interdependencies in TCL control. To address various application scenarios, we introduce optimization algorithms based on system identification (OptIden) and system simulation (OptSim) tailored to the latent variable representation. These approaches employ automatic differentiation and zeroth-order techniques, respectively, for efficient implementation. We evaluate the proposed method using a 90-zone apartment prototype, comparing its performance to traditional high-dimensional optimization. Results demonstrate that our approach effectively reduces control costs while achieving significantly higher computational efficiency.

\end{abstract}


\section{Introduction}\label{sec_Intro}
Thermostatically controlled loads (TCLs) account for approximately 40\% of total building energy consumption~\cite{0-1-energy1} and serve as key flexible resources for demand response by leveraging thermal inertia to regulate temperature~\cite{0-1-TCL}. However, optimizing TCLs in multi-zone buildings remains a significant challenge. Each zone within a building represents a distinct physical space equipped with a thermostat-controlled Heating, Ventilation, and Air Conditioning (HVAC) system to maintain thermal comfort~\cite{0-zone}. As the number of zones increases, the dimensionality of variables—such as indoor temperature, HVAC power, and external disturbances—grows, demanding greater computational resources for optimization. Additionally, multi-zone optimization requires modeling heat conduction dynamics to capture inter-zonal dependencies. These higher-order thermal couplings become increasingly complex as the number of zones expands, further complicating the quantification of thermal dynamics~\cite{0-2-thermal}.

Current research primarily focuses on advanced algorithms that rely on the original high-dimensional variables, but these methods often suffer from limitations in computational efficiency, model accuracy, and economic feasibility. Rule-based control strategies~\cite{0-rule} use empirical temperature curves to make decisions, offering simplicity but lacking the flexibility to adapt to varying thermal comfort requirements. Model predictive control (MPC)~\cite{0-2-thermal} incorporates more sophisticated control objectives but depends on accurate TCL dynamic models often unavailable to building operators. Meanwhile, model-free approaches such as reinforcement learning (RL)\cite{0-2-type-commerical-RL} can generate precise control decisions by training on large-scale temperature control scenarios, but such extensive training data is rarely practical to obtain in real-world buildings.

Recent research on dimensionality reduction has provided valuable insights for addressing high-dimensional optimization challenges. Dimensionality reduction involves transforming data from a high-dimensional space into a lower-dimensional latent space while preserving the key statistical properties of the original features~\cite{latent-model}. This transformation effectively reduces the complexity of the optimization problem. Some studies have explored dimension reduction techniques, such as model order reduction~\cite{1-1-fast,1-1-2012,1-1-other-review} and model aggregation~\cite{1-2-lump,1-2-Lu1,1-2-Lu2,1-2-aggregation}, to simplify high-dimensional optimization. However, the model order reduction typically assumes the availability of an accurate thermal dynamics model, which could be impractical for multi-zone buildings. \textcolor{black}{Model aggregation techniques achieve a balance between modeling accuracy and interoperability after obtaining the equivalent single-zone model, but there remains a lack of clarity on how to recover control decisions for temperature regulation from the reduced-dimensional solutions.}

\textit{How to formulate latent variables and utilize them to solve the optimization of multi-zone TCLs efficiently and accurately?} In this work, we answer this question from a data-driven perspective, introducing latent variable representations and dimension-reduced optimization. We propose optimization algorithms with both identification and simulation models 
to address different application scenarios, achieving higher computational efficiency and improved solution accuracy compared to using the original variables. The key contributions are as follows:
\begin{enumerate}
\item We develop a multi-task learning framework using auto-encoders (AEs) to learn latent states, actions, and disturbances. \textcolor{black}{A black-box model is integrated into the AEs to accurately capture time-coupled relationships among latent variables, effectively reducing dimensionality while ensuring the recoverability of original control signals.}

\item We propose optimization algorithms based on system identification (\textbf{OptIden}) and simulation (\textbf{OptSim}).
The OptIden reformulates the optimization into neural network (NN)-based structures, enabling gradient computation via automatic differentiation. The OptSim algorithm employs zeroth-order optimization to mitigate the effects of inaccuracies within the identification model.

\item We analytically derive the errors in latent variable representation and optimization results. We show, using empirical demonstrations, that the OptSim algorithm can achieve higher accuracy.
\end{enumerate}

The rest of this paper is organized as follows: Section II presents the literature review. Section III formulates the multi-zone TCLs optimization. Section IV provides the solution method. Section V analyzes the underlying errors within the method. Section VI conducts case studies and Section VII draws the conclusions.

\section{Literature Review}
\textcolor{black}{When participating in demand response, TCLs play an important role in providing various kinds of grid services, including peak shaving, frequency regulation, voltage regulation, and capacity reservation.
For example, HVAC systems in large residential communities were transferred into an equivalent energy storage model, providing substantial potential for reducing ramping rate and peak power by participating in a demand response scheme~\cite{new-2}.
Residential air conditioners were aggregated through a cyber-physical system to enable real-time control~\cite{new-1}, achieving frequency regulation services with minimal impact on compressor degradation and user comfort. Moreover, aggregate TCL models were integrated into microgrid operation~\cite{new-3}, enabling significant daily cost savings while meeting frequency regulation security levels.
Besides frequency regulation, TCLs were incorporated into a hierarchical energy management system for voltage control of a microgrid~\cite{new-5}. A complete set of simulations showed the effectiveness of the participation of TCLs in the frequency and voltage regulation.
As analyzed in~\cite{new-4}, TCLs could also be regarded as a contingency-type reserve resource by leveraging the statistical bounds on their exploitable flexibility.}

Extensive work has been done for the optimization of multi-zone TCLs with original variables. It can be roughly categorized into model-based and model-free methods. In model-based methods, thermal dynamics are quantified as a state-space model. Then, various optimization techniques are designed based on model structures. In~\cite{0-2-type-office}, the authors trained a neural network (NN) based thermal dynamics model with collected data in an office building. They transformed the NN forward process into piecewise linear constraints in the optimization problem. Optimal decisions were determined by solving the mixed integer linear programming (MILP). Authors in~\cite{0-2-type-commerical} considered the multi-source uncertainties in commercial buildings, and solved a real-time optimization problem based on Lyapunov optimization techniques. While much research assumed a linear resistor-capacitor (RC) based thermal dynamics model, the accuracy applied to multi-zone buildings has not been verified and guaranteed.
Due to difficulties in identifying thermal dynamics, some research conducts model-free learning and makes optimal decisions with measured data directly. In~\cite{0-2-type-commerical-RL}, the authors reformulated the optimization of multi-zone TCLs as a Markov game, and a multi-agent deep reinforcement learning algorithm was formulated with an attention mechanism. The model-free algorithm does not require prior knowledge of thermal dynamics. On the contrary, it trains a smart agent to learn from massive scenarios for temperature control. For multi-zone buildings, the required scenarios rise with the increasing zone number and are hard to access, thus making it hard to train a well-performing agent.

Instead of utilizing original and high-dimensional variables, some methods have been proposed to reduce the dimension of variables in the optimization. A summary of related work on dimension reduction is presented in Table~\ref{tab-review}. They are generally conducted in two ways: reducing the model complexity and improving the solving algorithm. Two kinds of methods have been proposed to reduce the model complexity, called order reduction and model aggregation. For order reduction, the high-order RC-based model~\cite{1-1-RCmodel} in multi-zone buildings is reduced into simpler structures through various techniques in control theory. For example, the balanced truncation method combined with controllability and observability analysis was applied in~\cite{1-1-fast} to obtain a high-fidelity order-reduced model. The authors in~\cite{1-1-2012} exploited the sparsity pattern of the nonlinear portion in the original model to build the simplified model. Other techniques (e.g., Kalman filter, Lyapunov balancing, and proper orthogonal decomposition) were also explored for model order reduction~\cite{1-1-other-review}. These methods can guarantee the theoretical accuracy of order-reduced models in arbitrary dimensions. However, a known full-order thermal dynamics model is required in advance, which is impractical in the actual application based on our prior analysis.

The second kind of model aggregation method is to obtain an equivalent single-zone model with aggregate variables. The model structure with aggregate variables is determined and model parameters are then identified. The authors in~\cite{1-2-lump} estimated a lumped model by weighted averaging of zone volume from a reference model. In~\cite{1-2-Lu1,1-2-Lu2}, a data-driven approach was proposed to aggregate thermal dynamics into a single-zone model. Aggregate variables were generated based on physical laws, and all model parameters were identified with auto-regression. In~\cite{1-2-aggregation}, a principled method based on the physics-informed Kalman filter was proposed to generate the single-zone model, where uncertain disturbances were also quantified. \textcolor{black}{Model aggregation methods can generate physically meaningful variables to improve the reliability of decisions, and they can balance between modeling accuracy and interoperability by adjusting orders of equivalent models. However, as shown in~\cite{1-2-lump}, the single-zone decisions after optimization were assumed as global control signals for all zones. How to disaggregate the equivalent single-zone signal to each zone remains unsolved in model aggregation-related techniques.}

Besides the model-related methods, other existing research conducts variable reduction with new solution algorithms. For example, the authors in~\cite{1-3-hierarch} proposed a hierarchical optimization algorithm where the whole variables were divided according to facility and zone levels; authors in~\cite{1-3-distributed} proposed a distributed optimization algorithm to transform the original problem into subproblems for zones and coupled areas.

\textcolor{black}{In light of the above context, we conclude two research gaps: 1) existing approaches on variable dimension reduction assume a known thermal dynamics model for reference, which could be unrealistic to be acquired for multi-zone buildings; 2) after obtaining reduced models/variables, there is little research on how to optimize TCLs with them, especially on how to recover the decisions of original variables for temperature control.}

\begin{table}[t]
\caption{Literature Comparisons}
\centering
\resizebox{0.9\linewidth}{!}{\begin{tabular}{c|cc|cc|c}
\hline
\multirow{2}{*}{Ref.} & \multicolumn{2}{c|}{Contribution} & \multicolumn{2}{c|}{Performance} & \multirow{2}{*}{Technique} \\ \cline{2-5}
                     & Model        & Optimization       & Feasibility     & Efficiency     &                            \\ \hline
\cite{1-1-fast,1-1-2012,1-1-other-review}                    &  \checkmark            &                    &                 &  \checkmark              & Order reduction            \\
\cite{1-2-lump,1-2-Lu1,1-2-Lu2,1-2-aggregation}                    &   \checkmark           &                    &                 &   \checkmark             & Aggregation             \\
\cite{1-3-hierarch}                    &              &     \checkmark               &   \checkmark              &                & Hierarchical                       \\
\cite{1-3-distributed}                  &              &   \checkmark                 &   \checkmark              &                & Distributed                \\ \hline
\textit{Proposed}                    &   \checkmark           &     \checkmark               &     \checkmark            &    \checkmark            &      Latent variables                      \\ \hline
\end{tabular}}
\label{tab-review}
\end{table}

\section{Problem Statement}\label{sec_Pro}
We present the TCL scheduling optimization for temperature control in multi-zone buildings. We consider a discrete time horizon $t\in\mathcal{T}$. We denote variables of all building zones in vector form per time step $v_t = [s_t, a_t, \delta_t]^\intercal$, where $s_t$ is the state (e.g., indoor temperature), $a_t$ is the action (e.g., HVAC electric power), and $\delta$ is the disturbances (e.g., outdoor temperature, solar radiation, and occupancy). The optimization is formulated as
\begin{subequations}\label{ori_pro}
\begin{align}
\min\limits_{a_t} & \sum\limits_{t \in {\cal T}} { \lambda_t \cdot {1}^\intercal a_{t} + P_{t} \cdot y_{t}^\intercal y_{t} } \label{ori_1} \\
\text{s.t.} ~ &{s_{t + 1}} = \mathcal{F}({s_t},a_t,{\delta_t}),~\forall t \label{ori_2}\\
&\underline{S}_t - y_t \le s_t \le \overline{S}_t + y_t,\, y_t \ge 0 ~\forall t \label{ori_4}\\
&\underline{A}_{t} \le a_{t} \le \overline{A}_{t},~\forall t \label{ori_6}
\end{align}
\end{subequations}
where in~\eqref{ori_1}, $\mathcal{T}$ is the optimization period; $\lambda_t$ and $P_t$ are the electricity price and penalty factor for temperature violation at each time $t$, respectively; $y_t$ is the vector of temperature violations; in~\eqref{ori_2}, $\mathcal{F}$ denotes the discrete-time state space model to represent the change of temperature; 
constraints in \eqref{ori_4}-\eqref{ori_6} set limits for the state and decision variables with bound values $\underline{S}_t,\overline{S}_t$ and $\underline{A}_t,\overline{A}_t$.

When applying~\eqref{ori_pro} for optimizing TCLs in multi-zone buildings, dimensions of $s,a$ and $\delta$ could increase linearly with the zone number. The high-dimensional variables in~\eqref{ori_pro} cause the following two issues:
\begin{enumerate}
\item Model error: The optimal decisions for temperature control in all zones are solved simultaneously when considering the coupled thermal dynamics~\cite{1-2-Lu1}. Given that case, identifying an accurate model of $\mathcal{F}$ is hard when capturing coupled and nonlinear factors from cross-zone and high-dimensional inputs~\cite{0-2-complex-model}.
\item Computational time: Solving the optimization problem with $v$ and complex and nonlinear $\mathcal{F}$ could be inefficient to obtain optimal decisions.
\end{enumerate}

\section{Methodology}\label{sec_Mod}
To solve the above problems, we propose to formulate latent variables whose dimensions are largely reduced. With latent variables, the model accuracy and the optimization efficiency could be improved. \textcolor{black}{In this section, we first introduce the developed multi-task learning based framework to represent latent variables from the original data. Then we present how to optimize TCLs with latent variables based on OptIden and OptSim algorithms. The sequential relationship between the representation and optimization steps is illustrated in Fig.~\ref{fig_work}.}

\begin{figure}[t]
\centering
\includegraphics[width=0.48\textwidth]{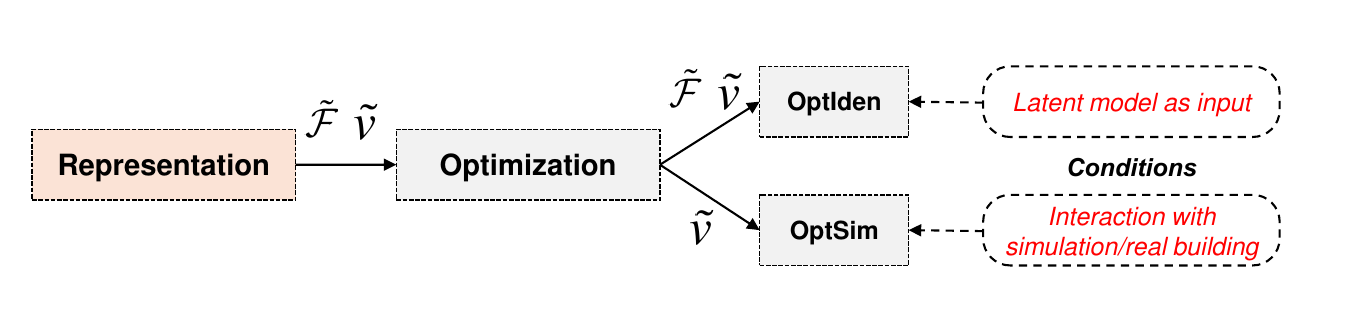}
\caption{\textcolor{black}{Workflow of the proposed methodology.}}
\label{fig_work}
\end{figure}

\subsection{Multi-task based Representation of Latent Variables}\label{sec_Mod1}
\subsubsection{Variable transformation}
We define the latent variables as $ \widetilde v = {[{\widetilde s},{\widetilde a},{\widetilde \delta}]^\intercal }$ whose dimension is much lower than ${v}$. In particular, we adopt AEs~\cite{chemE-1} as the unit model for variable transformation. The encoder $\mathcal E$ is defined as ${\widetilde v} = {\mathcal E}({v})$. Then the latent variables can recover to the original space with the decoder $\mathcal D$: ${v} = {{\mathcal E}^{-1}}({\widetilde v}) = {\mathcal D}({\widetilde v})$. Note that $\mathcal{E}$ and $\mathcal{D}$ are not necessarily to be linear and the transformation is independent, i.e.,
\begin{subequations} \label{AE_1}
\begin{align}
{[{\widetilde s},{\widetilde a},{\widetilde \delta}]^\intercal } & = {[\mathcal E_s({s}),{{\mathcal E}_a}({a}),{{\mathcal E}_\delta}({\delta})]^\intercal } \\
{[{s},{a},{\delta}]^\intercal } &= {[{\mathcal E}_s^{-1}({\widetilde s}),{\mathcal E}_a^{-1}({\widetilde a}),{\mathcal E}_\delta^{-1}({\widetilde \delta})]^\intercal } \nonumber \\
& = [{\mathcal D}_s({\widetilde s}),{{\mathcal D}_a}({\widetilde a}),{{\mathcal D}_\delta}({\widetilde \delta})]^\intercal 
\end{align}
\end{subequations}
where the subscripts ``$s, a, \delta$'' correspond to the state, action, and disturbance variables, respectively.

Then, we reformulate~\eqref{ori_pro} with latent variables:
\begin{subequations}\label{AE_new}
\begin{align}
\min\limits_{{\widetilde a}} & ~C(\mathcal{D}_s({\widetilde s}),\mathcal{D}_a({\widetilde a}),\mathcal{D}_\delta({\widetilde \delta})) \\
\text{s.t.} ~ &\mathcal{D}_s({\widetilde s}_{t+1}) = \mathcal F(\mathcal{D}_s({\widetilde s}_{t}),\mathcal{D}_a({\widetilde a}_{t}),\mathcal{D}_\delta({\widetilde \delta}_{t})) \label{AE_new_1} \\
& f(\mathcal{D}_s({\widetilde s}),\mathcal{D}_a({\widetilde a}),\mathcal{D}_\delta({\widetilde \delta})) \le 0 \label{AE_new_3}
\end{align}
\end{subequations}
where we denote the total objective function in~\eqref{ori_1} as $C$ and we utilize $f$ to represent the inequality constraints in~\eqref{ori_4}-\eqref{ori_6}. With latent variables as the inputs, they are transferred into $\widetilde C(\widetilde v)$ and $\widetilde f(\widetilde v)$, respectively.

In this new form of optimization, the key step is to identify the function~\eqref{AE_new_1}, which represents the dynamics within the latent space. We rewrite~\eqref{AE_new_1} as:
\begin{equation} \label{AE_equa_1}
\begin{split}
{\widetilde s}_{t+1} &= \mathcal{E}_s [ \mathcal{D}_s({\widetilde s}_{t+1}) ] \\
& = \mathcal{E}_s [ \mathcal F(\mathcal{D}_s({\widetilde s}_{t}),\mathcal{D}_a({\widetilde a}_{t}),\mathcal{D}_\delta({\widetilde \delta}_{t}))]
\end{split}
\end{equation}
which can be approximated with a latent state model $\mathcal {\widetilde F}$:
\begin{equation} \label{AE_equa_2}
{\widetilde s}_{t+1} = \mathcal {\widetilde F} ({\widetilde s}_{t},{\widetilde a}_{t},{\widetilde \delta}_{t})
\end{equation}

The \textbf{dimension-reduced optimization} with latent variables is
\begin{subequations}\label{AE_new_new}
\begin{align}
\min\limits_{{\widetilde a}} & ~\widetilde C({\widetilde s},{\widetilde a},{\widetilde \delta}) \\
\text{s.t.} ~ & \eqref{AE_equa_2} \\
& \widetilde f({\widetilde s},{\widetilde a},{\widetilde \delta}) \le 0 \label{AE_new_new_3}
\end{align}
\end{subequations}

In this problem, multiple tasks on identifying $\mathcal{E}$, $\mathcal{\widetilde F}$, and $\mathcal{D}$ should be conducted for variable transformation, dynamics modeling, and variable recovery, respectively.

\subsubsection{Training algorithm}
We advocate a multi-task learning based method to identify the required functions jointly. The proposed framework is presented in Fig.~\ref{fig_med_model}.
\begin{figure*}[t]
\centering
\includegraphics[width=0.8\textwidth]{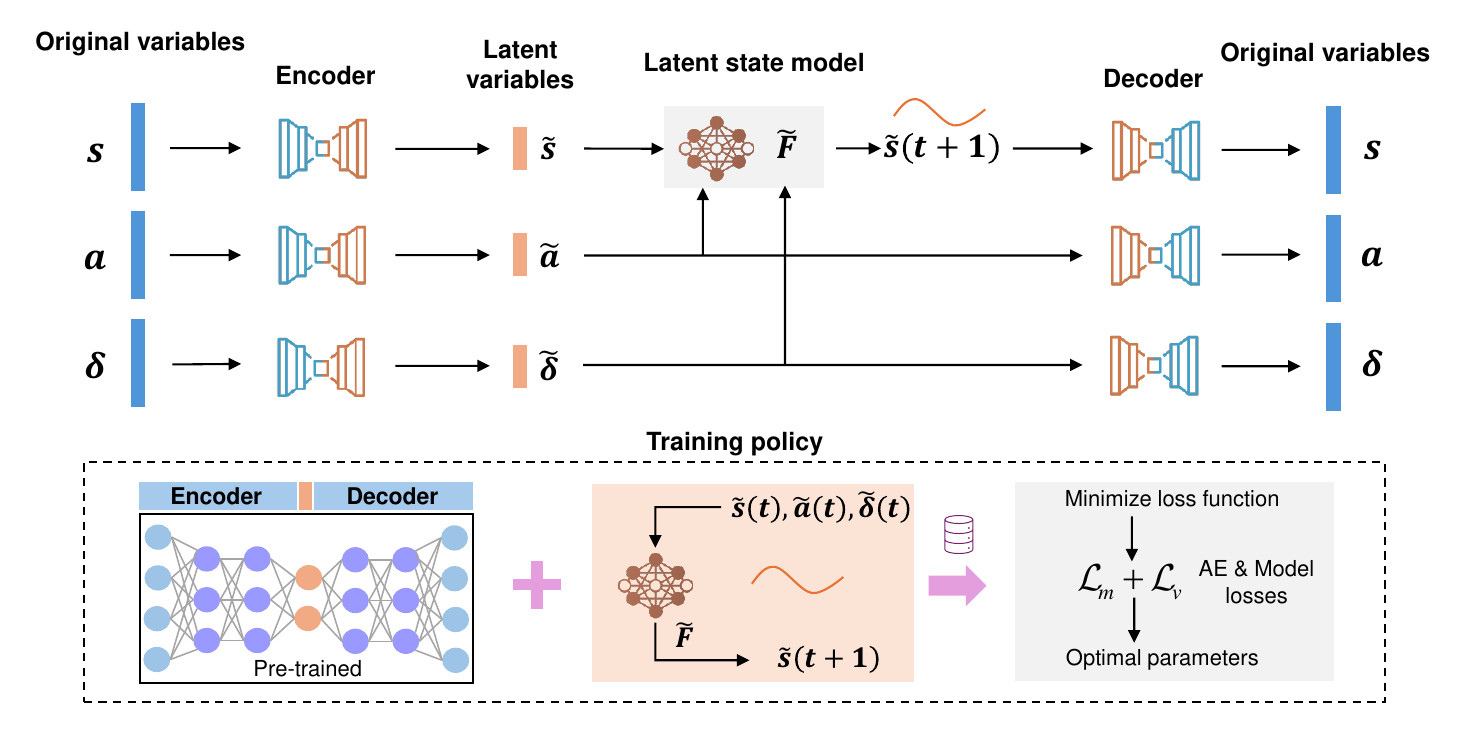}
\caption{Proposed framework for the representation of latent variables.}
\label{fig_med_model}
\end{figure*}

During the training process, the inputs of measured data are collected within the training period $\mathcal{T}^\text{train}$. \textcolor{black}{The collected data can potentially characterize the differences in thermal dynamic properties (e.g., thermal inertia) among different rooms. Therefore, by accurately characterizing the raw data, our trained model can effectively capture the variations in thermal inertia across rooms with latent variables.} The multi-task loss function $\mathcal L$ for training is formulated as
\begin{equation} \label{AE_loss_1}
\begin{split}
{\mathcal L} = & \omega {\mathcal L}_{m} + (1-\omega) {\mathcal L}_{v} \\
= & \omega \sum_{t \in \mathcal{T}^\text{train}} {\left \| \mathcal D_s[{\mathcal {\widetilde F}}({\mathcal E}({v}_t))] - {s}_{t+1} \right \|_2^2} \\
+ & (1-\omega) \sum_{t \in \mathcal{T}^\text{train}} { \left \| {\mathcal D}\left[ {{\mathcal E}({v}_t)} \right] - {v}_t \right \|_2^2}  \\
\end{split}
\end{equation}
where the first term ${\mathcal L}_{m}$ is to minimize errors between outputs of $\mathcal{\widetilde F}$ and measured states; the second term ${\mathcal L}_{v}$ is to minimize errors of encoders and decoders for variable transformation; $\omega$ is the weighting factor.

We further infer the gradients of $\mathcal L$ over all parameters to be learned as
\begin{subequations}\label{AE_loss_2}
\begin{align}
& \frac{{\partial {\cal L}}}{{\partial \widetilde {\cal F}}} = \frac{{\partial {{\cal L}_m}}}{{\partial {{\widetilde s}_{t + 1}}}}\frac{{\partial {{\widetilde s}_{t + 1}}}}{{\partial \widetilde {\cal F}}}\\
&\frac{{\partial {\cal L}}}{{\partial {{\cal D}_s}}} = \frac{{\partial {{\cal L}_m}}}{{\partial {{\cal D}_s}}} + \frac{{\partial {{\cal L}_v}}}{{\partial {{\cal D}_s}}},~\frac{{\partial {\cal L}}}{{\partial {{\cal D}_{a,\delta}}}} = \frac{{\partial {{\cal L}_v}}}{{\partial {{\cal D}_{a,\delta}}}}\\
&\frac{{\partial {\cal L}}}{{\partial {\cal E}}} = \frac{{\partial {{\cal L}_m}}}{{\partial {{\widetilde s}_{t + 1}}}}\frac{{\partial {{\widetilde s}_{t + 1}}}}{{\partial {{\widetilde v}_t}}}\frac{{\partial {{\widetilde v}_t}}}{{\partial {\cal E}}} + \frac{{\partial {{\cal L}_v}}}{{\partial {{\widetilde v}_t}}}\frac{{\partial {{\widetilde v}_t}}}{{\partial {\cal E}}}
\end{align}
\end{subequations}
which can be easily calculated by automatic differentiation in existing NN techniques~\cite{autograd}. Then the widely used back-propagation (BP) algorithm can be adopted for multi-task training.

\begin{remark}[\textcolor{black}{\textit{Independent AEs for variable representation}}]
\textcolor{black}{In Fig.~\ref{fig_med_model}, we adopt three independent AEs to represent latent variables of state, action, and disturbances, respectively. Compared with adopting only one AE for latent variable representation, the proposed framework has the following advantages: 1) it effectively isolates three kinds of variables during the representation of latent variables, enabling us to train a reduced-order model that preserves physical interpretability; 2) it offers adaptive dimensionality configuration for latent variables, enabling flexible representation of thermal dynamics. Specifically, complex time-varying disturbance variables are encoded into higher-dimensional latent spaces compared to other variables, thereby ensuring precise characterization of their dynamic behaviors.}
\end{remark}

\subsection{Dimension-reduced Optimization with Latent Variables}\label{sec_Mod3}
In this subsection, we introduce how to solve the dimension-reduced optimization with identified latent variables. We begin with the proposed OptIden algorithm where the identified latent model is used as the state-space model. On this basis, we explore the scenario where a OptSim algorithm can be formulated with feedback from the real environment.

\subsubsection{OptIden algorithm}
\begin{figure}[t]
\centering
\includegraphics[width=0.48\textwidth]{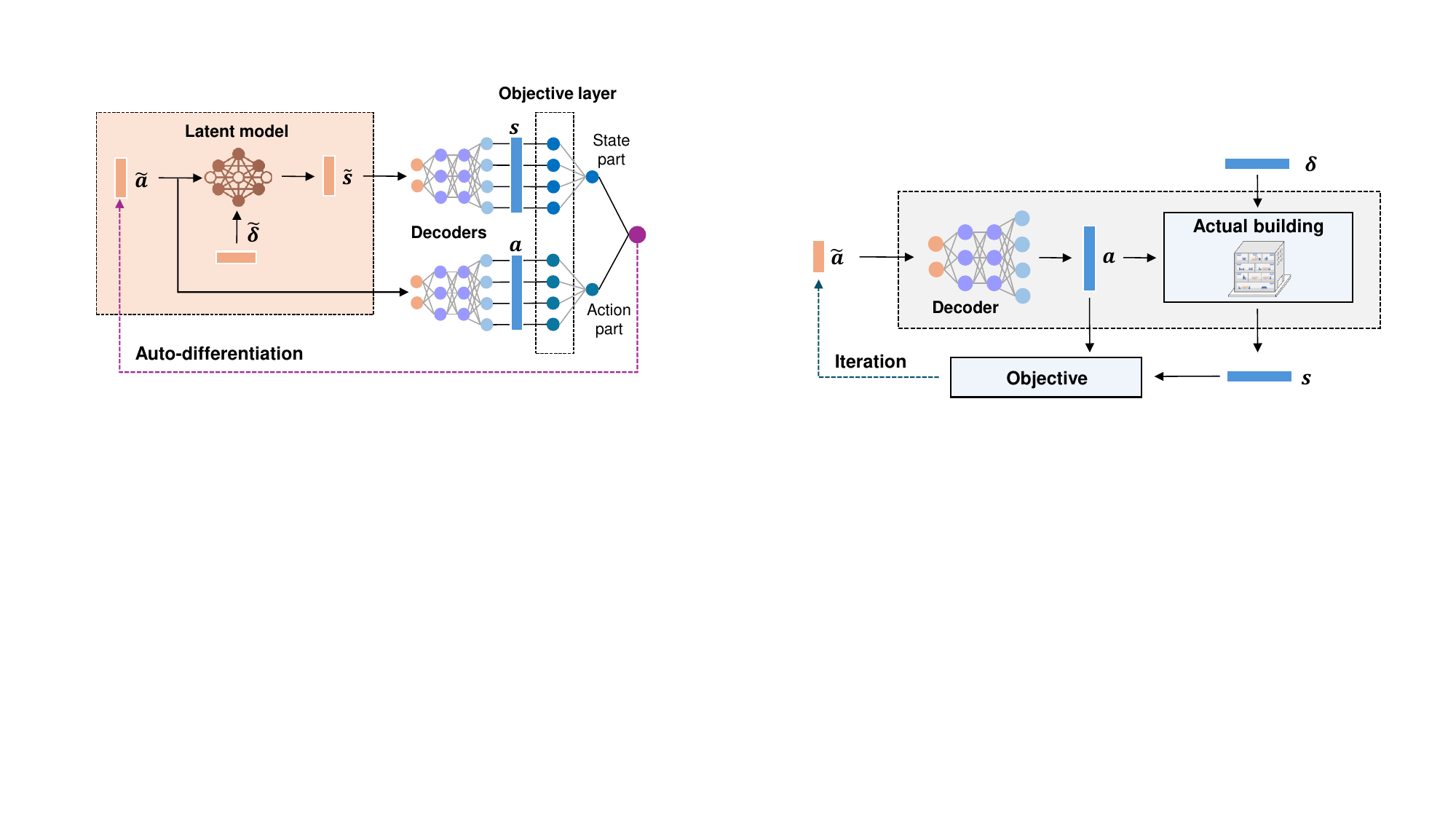}
\caption{Proposed OptIden algorithm with latent variables.}
\label{fig_med_opt}
\end{figure}

The proposed OptIden algorithm is presented in Fig.~\ref{fig_med_opt}. In particular, we rewrite~\eqref{AE_new_new} by adding the constraints as penalty terms, and the total objective function $\overline C$ is denoted as:
\begin{subequations}\label{opt_1}
\begin{align}
&\min \left \{ \overline C | \text{s.t.} ~ \eqref{AE_equa_2},~\forall t \in \mathcal{T} \right \} \\
&\overline C = \widetilde C + \rho \cdot ([\widetilde f, 0]_+)^2
\end{align}
\end{subequations}
where $\rho$ are weighting factors and $[\widetilde f, 0]_+$ returns $\widetilde f$ if $\widetilde f > 0$ and 0 otherwise.

Due to the unknown model structure in \eqref{AE_equa_2}, we design a gradient descent algorithm to solve the problem heuristically. Specifically, we accelerate the gradient derivation by reformulating \eqref{AE_equa_2} into NNs and facilitating the automatic differentiation. In Fig.~\ref{fig_med_opt}, it can be seen that the required gradients are formulated as
\begin{equation}\label{model_grd_1}
\frac{\partial \overline C}{\partial {\widetilde a}} = \frac{\partial \overline C}{\partial {a}} \cdot \frac{\partial {a}}{\partial {\widetilde a}} + \frac{\partial \overline C}{\partial {s}} \cdot \frac{\partial {s}}{\partial {\widetilde a}}
\end{equation}

For the latter terms, $\frac{\partial {a}}{\partial {\widetilde a}}$ and $\frac{\partial {s}}{\partial {\widetilde a}}$ are easily obtained through $\mathcal{D}_a$, $\mathcal{D}_s$, and $\mathcal{\widetilde F}$ in the proposed latent space. For the former terms, we design an objective layer based on calculating $\overline C$.

Note that the calculation of $\overline C$ includes the following general terms:
\begin{equation}\label{obj_terms}
\left \{ y=\omega x;~y=\omega x^2; ~y=w[x-{x_0},0]_+^2 \right \}
\end{equation}

The forward process in a standard NN is:
\begin{equation}\label{nn_terms}
z_l = \varphi (W \cdot {z_{l - 1}} + b),~\forall l
\end{equation}
where $z_l$ is the output of the $l$-th layer; $W,b$ are weight matrix and bias; $\varphi$ is the activation function. As shown in Table~\ref{tab-layer}, each term in~\eqref{obj_terms} can be transformed into~\eqref{nn_terms}.
\begin{table}[t]
\caption{Components in Objective Layer}
\centering
\resizebox{0.8\linewidth}{!}{\begin{tabular}{c|ccccc}
\hline
\multirow{2}{*}{\diagbox{Obj.}{NN.}} & \multicolumn{5}{c}{ $z_l = \varphi (W \cdot {z_{l - 1}} + b)$ } \\ \cline{2-6}
           & $\varphi$   & $W$   & $b$   & $z_{l - 1}$   & $z_{l}$   \\ \hline
$y=wx$       &  /  & $w$   &  0  &  $x$  & $y$  \\
$y=wx^2$      &  /  &  $\sqrt w$  & 0   & $x$  & $\sqrt y$  \\
$y=w[x-{x_0},0]_+^2$ & ReLU   & $\sqrt w$   & 0   & $x-x_0$   & $\sqrt y$   \\ \hline
\end{tabular}}
\label{tab-layer}
\end{table}

After gradient derivation through NNs, the overall OptIden method is presented in Algorithm 1. \textcolor{black}{In the algorithm, we adopt adaptive step sizes to balance the solving performance and convergence. Besides, we set conditions for step size update and convergence with the indicator function $\mathbb{I}[x]$, which returns 1 if the statement $x$ is true and 0 otherwise. In this way, if the cost function increases for $K_1$ continuous iterations, the step size is halved; if the absolute change in the cost function remains below 0.01 for $K_2$ continuous iterations, the algorithm converges early. The values of thresholds $K_1$ and $K_2$ cause dependent influence on the algorithm's convergence, and we determine them by grid search jointly.}
\begin{algorithm}[t]
\caption{OptIden Algorithm}\label{alg:model}
\KwIn{$\mathcal{T}$: optimization period; ${\delta}_t$: measured value of all disturbances at each time $t$; ${s}_0$: initial true value of states; $\mathcal{K}$: maximum iterations; $\eta$: step size.}
\KwOut{${\widetilde a}^*$: Variable of optimal decisions.}
Initialize decision variables ${\widetilde a}_{t,0}, t \in \mathcal{T}$ \;
Obtain the true value of latent disturbances by \\
${\widetilde \delta}_{t} = \mathcal{E}_\delta({\delta}_{t}) $\;
\For{each iteration $k \in \mathcal{K}
$}{
Obtain states ${\widetilde s}_{t,k}$ based on~\eqref{AE_equa_2} with inputs ${\widetilde a}_{t,k}, {\widetilde \delta}_{t}$ and then determine ${s}_{t,k} = \mathcal{D}_s({\widetilde s}_{t,k})$ \;

Calculate objective values $\overline C_k$ with inputs ${s}_{t,k}, {a}_{t,k}$\;

Determine gradients of each feature ${\widetilde a}_{k}$ by~\eqref{model_grd_1}, and update ${\widetilde a}_{k+1}$ by \\
${\widetilde a}_{k+1} = {\widetilde a}_{k} - \eta \cdot \nabla_{{\widetilde a}_{k}} \overline C_k$ \;
\textcolor{black}{
\If{$\sum\limits_{k'=k-K_1+1}^{k}{\mathbb{I}\left [\overline C_{k'+1} - \overline C_{k'} > 0 \right ] }\ge K_1$}{$\eta \gets \frac{1}{2} \cdot \eta$}
\If{$\sum\limits_{k'=k-K_2+1}^{k}{\mathbb{I}\left [|\overline C_{k'+1} - \overline C_{k'}| \le 0.01\right ]} \ge K_2$}{
\textbf{break}}}
}
\end{algorithm}

\subsubsection{OptSim algorithm}
As we will discuss in Section V, there are potential errors in the identified latent model that could lead to biased optimization solutions. Here we present an option to alleviate underlying errors, i.e., the OptSim algorithm, by interacting with the real environment under the following assumption.

\begin{assumption}[\textit{Fully observable state}]\label{assumption1}
We assume the optimization can immediately observe the updated system state $s_{t+1}$ following the application of new actions $a_t$ and disturbances $\delta_t$. In the context of building control, this implies that the building operator can instantly measure the updated indoor temperatures across all zones.
\end{assumption}

Note that Assumption~\ref{assumption1} holds true under the following conditions: 1) The building operator has access to a high-fidelity, real-time simulation model of the building, such as a digital twin; or 2)
The building's conditions, including occupancy and ambient factors, are relatively stable, allowing the operator to assume steady-state operation and observe temperature updates over an extended time horizon.

As shown in Fig.~\ref{fig_med_freeopt}, the major idea of the OptSim is to solve $\min \left \{ \overline C | \text{s.t.} ~ \eqref{ori_2},~\forall t \in \mathcal{T} \right \}$. The difference from~\eqref{opt_1} is that we only use the trained decoder $\mathcal{D}_a$ to obtain ${a}$, and we directly quantify $\overline C$ together with the outputs of simulation from $\mathcal{F}$. This procedure bypasses a trained model to avoid the influence of model errors.
\begin{figure}[t]
\centering
\includegraphics[width=0.48\textwidth]{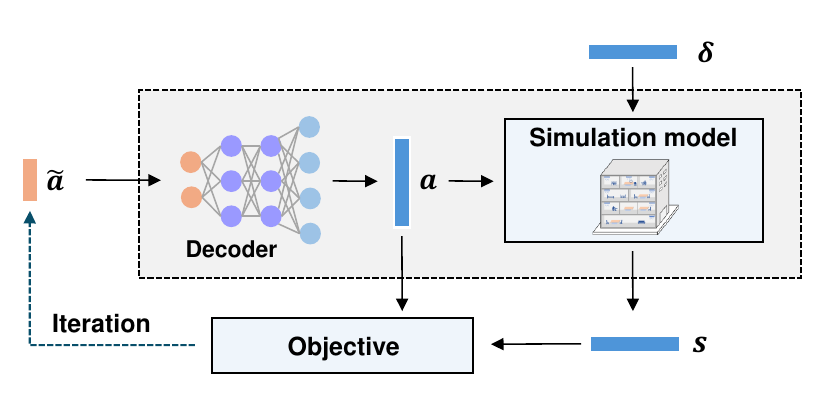}
\caption{Proposed OptSim algorithm with latent variables.}
\label{fig_med_freeopt}
\end{figure}

However, unlike the automatic differentiation adopted in the OptIden algorithm, the gradients in the OptSim algorithm can not be inferred because of the unknown $\mathcal{F}$. To that end, we regard both $\mathcal{D}_a$ and unknown $\mathcal{F}$ as an overall black-box function, i.e., $\mathcal{F}'({\widetilde a}) = \mathcal{F}(\mathcal{D}_a({\widetilde a}),\delta)$. Then we adopt the zeroth-order optimization to update ${\widetilde a}$. The update at each $k$-th iteration is presented as
\begin{subequations}\label{zero}
\begin{align}
&{\widetilde a}_{k+1} = {\widetilde a}_{k} - \eta \cdot {\xi}_k + \alpha ({\widetilde a}_{k} - {\widetilde a}_{k-1}) \label{zero_1} \\
&{\xi}_k = \frac{1}{r} \left [\overline C_k({\widetilde a}_{k} + r{u}_k;\mathcal{F}')-\overline C_k({\widetilde a}_{k};\mathcal{F}')\right ]{u}_k
\end{align}
\end{subequations}
where a two-point zeroth order optimization is adopted~\cite{Zeroth-method}. ${\xi}$ denotes the update term, representing the difference of measured objective values to replace gradients; $r$ is the radius to control the difference of input vector; ${u}_k \sim \mathcal{U}([-1,1])$ is the perturbed vector sampled from the uniform distribution $\mathcal{U}$ at each $k$-th iteration. An additional momentum term with weighting factor $\alpha$ is adopted in~\eqref{zero_1} to improve convergence.

\textcolor{black}{Note that, as a model-free technique for real-time control in a dynamic and complex environment, reinforcement learning (RL) has emerged as a promising tool for TCL scheduling by interacting with the simulation platform. The simulation would define the state, action, and reward function for RL training. In this context, our method shows high potential to address complementary challenges when integrating with RL. More details on related future research can be found in Section VII.}

\begin{remark}[\textit{Zeroth order optimization with latent variables}]
In Fig.~\ref{fig_med_freeopt}, we regard the newly defined $\mathcal{F}'$ as the overall simulation model. In this way, we can only update ${\widetilde a}$ instead of the original variable ${a}$. As we will show in Section VI, the performance of zeroth order optimization is influenced by the dimension of input variables. Thus, the major advantage of adopting the latent variable is that it improves convergence by reducing the variable dimension.
\end{remark}

\section{Latent Error Analysis}\label{sec_Mod4}
The data-driven modeling of latent variables in OptIden could lead to unavoidable errors (e.g., in $\mathcal{D}$ and $\mathcal{\widetilde F}$). These model errors could cause biased results in the following optimization process. In this section, we analytically present such error propagation in the OptIden algorithm. Then we show how the proposed OptSim algorithm helps alleviate such errors.

\subsection{Errors of OptIden Algorithm}
The errors to be analyzed are shown in Fig.~\ref{fig_err_begin}. We will first quantify the error during latent variable modeling, then we will discuss the influence of model errors on optimization results.
\begin{figure}[t]
\centering
\includegraphics[width=0.48\textwidth]{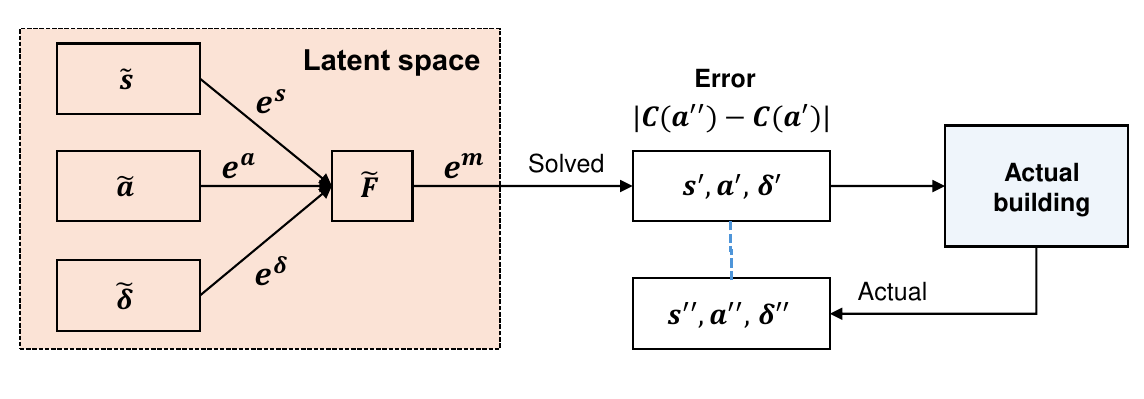}
\caption{Error analysis within the OptIden algorithm.}
\label{fig_err_begin}
\end{figure}

\subsubsection{Modeling of latent variables}
\begin{definition}
We denote the errors when formulating the latent state, action, and disturbance variables as ${e}^s,{e}^a$, and ${e}^\delta$, respectively, which are defined as
\begin{subequations}\label{err_1}
\begin{align}
& {e}^s = {s} - {{\mathcal D}_s}({\widetilde s}) \\
& {e}^a = {a} - {{\mathcal D}_a}({\widetilde a}) \\
& {e}^\delta = {\delta} - {{\mathcal D}_\delta}({\widetilde \delta})
\end{align}
\end{subequations}

Then the model error is defined as
\begin{equation}
{e}^m = {s} - {{\mathcal D}_s}(\mathcal {\widetilde F})
\end{equation}
which is the difference between the real state and the output of the identified latent model.
\end{definition}

\begin{proposition}[\textit{Latent model error analysis}]\label{the_1}
The model error ${e}^m$ depends on ${e}^s$ and is independent of ${e}^a$ and ${e}^\delta$. 
\end{proposition}

\begin{proof}
We first define the learning error of $\mathcal{\widetilde { F}}$ as
\begin{equation}
{{\widetilde e}^m} = {\widetilde s} - \mathcal{\widetilde { F}}({\widetilde a},{\widetilde \delta})
\end{equation}
which measures the difference between the model output and the latent state variable.

Then we can build the relationship between ${e}^m$ and ${e}^s$ by
\begin{equation}
\begin{split}
{e}^m &= {s} - \mathcal{D}_s(\mathcal{\widetilde F}) = {s} - \mathcal{D}_s({\widetilde s} - {\widetilde e}^m) \\
& = {s} - \mathcal{D}_s({\widetilde s}) + \mathcal{D}_s({\widetilde s}) -\mathcal{D}_s({\widetilde s} - {\widetilde e}^m) \\
& = {e}^s + \epsilon({{\widetilde e}^m};{\widetilde s})
\end{split}
\end{equation}
where $\epsilon$ is a residual term dependent on the variable $\widetilde s$ and the error ${\widetilde e}^m$.
\end{proof}

\subsubsection{Optimization results}
Given the latent model, we continue to quantify the errors in the optimization results.

\begin{definition}
After solving~\eqref{opt_1}, the results of the original variables can be obtained by
\begin{subequations}
\begin{align}
&{a}' = {{\mathcal D}_a}({\widetilde a})\\
&{s}'_{t+1} = {{\mathcal D}_s}(\mathcal{\widetilde F}({\widetilde s}_{t},{\widetilde a}_t,{\widetilde \delta}_t)),~\forall t
\end{align}
\end{subequations}
where ${a}'$ and ${s}'$ means solutions of~\eqref{opt_1}.

\textcolor{black}{Note that the optimization problem of~\eqref{ori_pro} includes soft and hard constraints for the state~\eqref{ori_4} and action~\eqref{ori_6} limits, respectively. Thus, after transforming the hard constraint~\eqref{ori_6} into the penalty term in~\eqref{opt_1}, it may cause violations regarding power bounds when the penalty is not zero. A strategy is formulated to regulate the violated part and thus obtain the actual values:}
\textcolor{black}{\begin{subequations}
\begin{align}
&{{a}^{{''}}} = \text{Proj}({{a}^{{'}}}) = \left\{ \begin{array}{l}
{{a}^{{'}}};~\underline a \le {{a}^{{'}}} \le \overline a\\
\underline a;~{{a}^{{'}}} < \underline a\\
\overline a;~{{a}^{{'}}} > \overline a
\end{array} \right. \\
&{s}''_{t+1} = \mathcal{F}({s}''_{t},{ a}''_t,\delta_t),~\forall t
\end{align}
\end{subequations}
where $\text{Proj}()$ is the projection function; ${a}''$ is the actual action variable after projection and ${s}''$ is the actual state variable.}

Taking $({s}',{s}')$ as an example, and denoting the zone set as $\mathcal{Z}$, the objective in~\eqref{ori_1} can be transformed into
\begin{equation}\label{math_2_defi}
\begin{split}
&C({a}',{s}') \\
& = \sum\limits_{t \in {\mathcal T}} {\sum\limits_{i \in {\mathcal Z}} {\lambda_t a'_{i,t} + P_{i,t}([s'_{i,t} - \overline s_{i,t}]_+^2 + [\underline s_{i,t} - s'_{i,t}]_+^2) }}
\end{split}
\end{equation}
and $C({a}'',{s}'')$ can be calculated in the same way.
\end{definition}



\begin{proposition}[\textit{Latent optimization error analysis}]\label{propo_opt}
The optimization error $C(a'',s'') -C(a',s')$ has a nonlinear relationship with $e^a$ and $e^m$ and is independent of $e^\delta$. 
\end{proposition}

\begin{proof}



It can be seen from~\eqref{math_2_defi} that the objective values are the summation of results at all time slots and building zones. For convenience, we only discuss the single-period result and ignore the indices $i,t$.

We first quantify the errors between solved and actual variables. For $a^{{''}}$ and $a^{{'}}$, the error $\widetilde e^a$ is calculated as
\begin{equation}\label{error_A}
\begin{split}
{\widetilde e^a}(a') &= {a^{{''}}} - {a^{{'}}} \\
&= \left\{ \begin{array}{l}
0;~\underline a \le {a^{{'}}} \le \overline a\\
\underline a - {a^{{'}}};~{a^{{'}}} < \underline a\\
\overline a - {a^{{'}}};~{a^{{'}}} > \overline a
\end{array} \right.
\end{split}
\end{equation}

Then the error $\widetilde e^s$ for $s^{{''}}$ and $s^{{'}}$ is
\begin{equation}\label{error_S}
\begin{split}
\widetilde e^s &= {s^{{''}}} - {s^{{'}}} \\
& = {\mathcal F}({a^{{''}}},\delta) - {s^{{'}}}\\
& = {\mathcal F}({a^{{'}}} + {\widetilde e^a},\delta) - {s^{{'}}}\\
& = {\mathcal F}(a - {e^a} + {\widetilde e^a},\delta) - s^{{'}}\\
& = {\mathcal F}(a,\delta) + \epsilon'({\widetilde e^a} - {e^a};a) - {s^{{'}}}\\
& = {e^m} + \epsilon'(\widetilde e^a - e^a;a)
\end{split}
\end{equation}
where $\epsilon'$ is the residual term related to variable $a$ and the error $(\widetilde e^a-e^a)$. Based on~\eqref{error_A} and~\eqref{error_S}, we know $\widetilde e^s$ is related to $(e^m,e^a,a')$.

According to~\eqref{math_2_defi}, the error between the solved and actual objectives can be represented as
\begin{equation}\label{math_err_act}
\begin{split}
\widetilde e &= C({a^{{''}}},{s^{{''}}}) - C({a^{{'}}},{s^{{'}}}) \\
&= \lambda \cdot {\widetilde e^a} + P([(s'+\widetilde e^s) - \overline s]_+^2 + [\underline s - (s'+\widetilde e^s)]_+^2) \\
&- P([s' - \overline s]_+^2 - [\underline s - s']_+^2)
\end{split}
\end{equation}
where the error related to the temperature penalty part is explained in Fig.~\ref{fig_err_piece}.
\begin{figure}[t]
\centering
\includegraphics[width=0.48\textwidth]{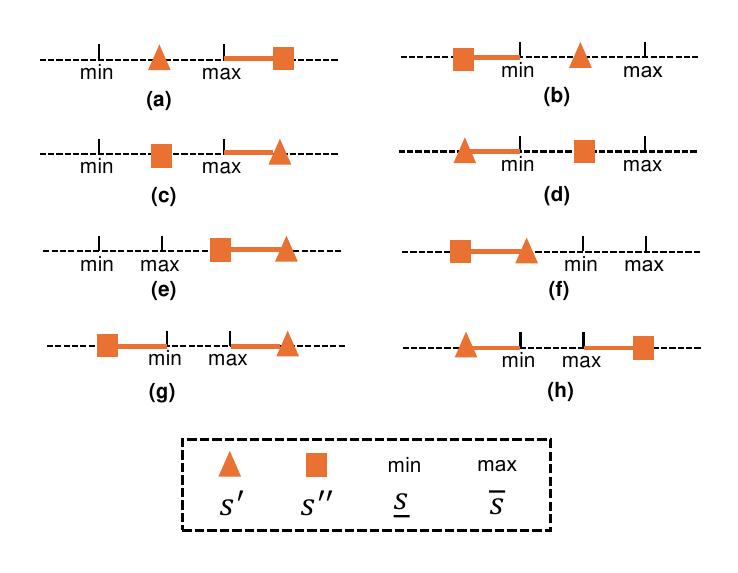}
\caption{Errors in the temperature-related objectives.}
\label{fig_err_piece}
\end{figure}


Finally, it can be found that $e'$ is related to $(\widetilde e^a,\widetilde e^s)$, and then $(e^m,e^a,a')$, i.e., the error in the optimization result is mainly influenced by the representation errors of latent action and state variables.
\end{proof}

\begin{remark}[\textcolor{black}{\textit{Mitigation of modeling and optimization errors}}]
\textcolor{black}{Based on the above analysis, we provide empirical strategies to reduce errors in latent modeling and optimization. Specifically, achieving an accurate latent model requires particular attention to reducing state variable errors. Practical approaches include fine-tuning the AE's hyperparameters (e.g., the dimension of latent variables), fine-tuning the model's hyperparameters, and improving the training strategy (e.g., the maximum epoch and learning rate). Furthermore, the learning errors of models and decision variables should be highlighted to reduce optimization errors and obtain reliable decisions. Besides, the rational setting of threshold values in Algorithm 1 could help to reach suboptimal solutions.}
\end{remark}

\subsection{Errors of OptSim Algorithm}
We will analyze the errors in the OptSim algorithm and show how it alleviates the optimization errors in the OptIden algorithm.
Because only the decoder $\mathcal{D}_a$ is adopted in the OptSim algorithm, there is only $e^a$ in the latent space. Similarly, the error $\widetilde e^a$ in~\eqref{error_A} also exists considering the violation risk of action variables. Under Assumption 1, the OptSim algorithm can use a simulation model as $\mathcal{F}$ to output the actual states, so the error $\widetilde e^s$ can be represented by
\begin{equation}\label{error_free_S}
\begin{split}
\widetilde e^s
& = {\mathcal F}({a^{{''}}},\delta) - {\mathcal F}({a^{{'}}},\delta)\\
& = {\mathcal F}({a^{{'}}} + {\widetilde e^a},\delta) - {\mathcal F}({a^{{'}}},\delta)\\
& = \epsilon'(\widetilde e^a;a')
\end{split}
\end{equation}
Hence $\widetilde e^s$ in the OptSim algorithm depends only on $\widetilde e^a$ but has no relationship with $e^a$. Considering that the penalty term generally has small errors after solving, the solution of the OptSim algorithm has marginal error in optimization results.

\begin{remark}[\textit{Error comparison}]
Based on the above analysis, the OptSim algorithm can effectively reduce the errors in the optimization results by avoiding identified latent modeling errors. For the OptIden algorithm, improving the modeling accuracy is critical for decision-making; for the OptSim algorithm, a prerequisite is the interaction with the simulation model, while the effectiveness of the zeroth order optimization will be tested in the case study.
\end{remark}

\section{Case Studies}\label{sec_Case}
In this section, we verify the performance of the proposed method with a multi-zone building case. We compare the optimization results for TCL scheduling with methods based on original variables. We also discuss some key issues during the solving process, involving time efficiency, convergence, data robustness, etc. Source code, input data, and solving results are available on Github\footnote{\url{https://github.com/hkuedl/Latent-Variable-based-Optimization}}.

\subsection{Simulation Setup}
\subsubsection{Data preparation}
In our case, we adopt a 90-zone apartment prototype provided by the U.S. Department of Energy (DOE) for case simulation~\cite{case-buildings}. \textcolor{black}{Furthermore, by scaling the number of zones per floor to four times while tripling the number of floors, we assume an extended building prototype with a 12-fold increase in total zone number (i.e., 1080 zones). Correspondingly, we augment the simulated data of the 90-zone building to obtain the data for the study of the 1080-zone case. The illustration of the adopted building case and the extension process is explained in Fig.~\ref{fig_building}.} In particular, we collect available data including indoor temperature (℃), outdoor temperature (℃), solar radiation power (kW), internal occupancy (kW), and HVAC cooling power (kW) of each zone. Based on these data, we build a black-box surrogate model, assuming it as the ground-truth building and also the simulation model for OptSim. For the learning-based representation of latent variables, the training and test periods of collected data are set as 01/06-31/07 and 01/08-31/08, respectively. We set a day-ahead optimization for TCL scheduling with a 15-minute time interval. \textcolor{black}{We set the measured average temperature series as the baseline curve and set symmetric and time-varying upper/lower bounds as the comfort level constraints. For example, the temperature comfort bounds on the first day of the training period are presented in Fig.~\ref{fig_comfort}.} More detailed parameter settings are presented in Table~\ref{tab-param}. All algorithms are executed on a computer with a 3.40 GHz Intel Xeon(R) CPU with 160 GB of RAM.
\begin{figure}[t]
\centering
\includegraphics[width=0.45\textwidth]{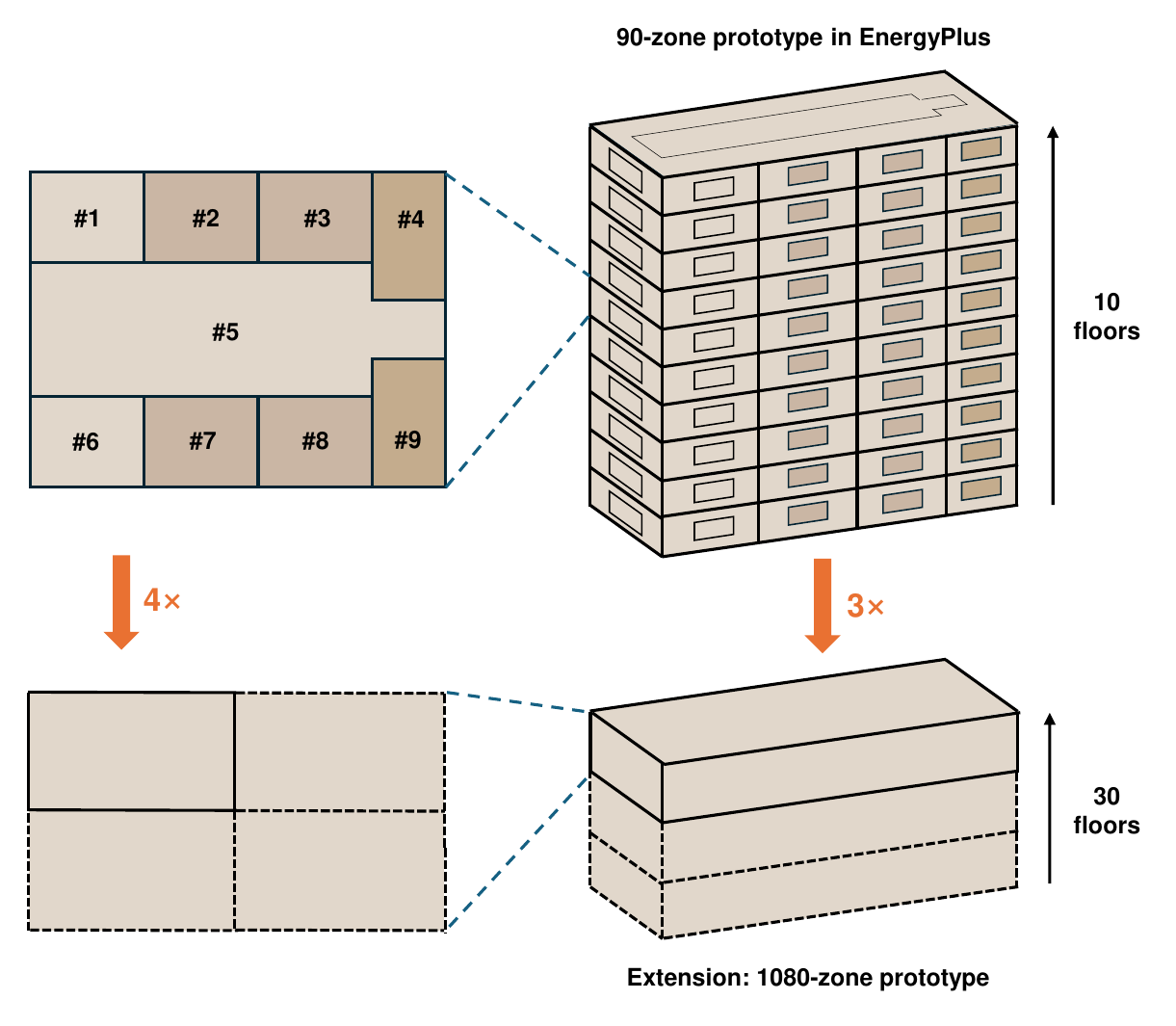}
\caption{\textcolor{black}{Building prototypes for case study (the 90-zone building in EnergyPlus and the assumed 1080-zone building).}}
\label{fig_building}
\end{figure}
\begin{figure}[t]
\centering
\includegraphics[width=0.45\textwidth]{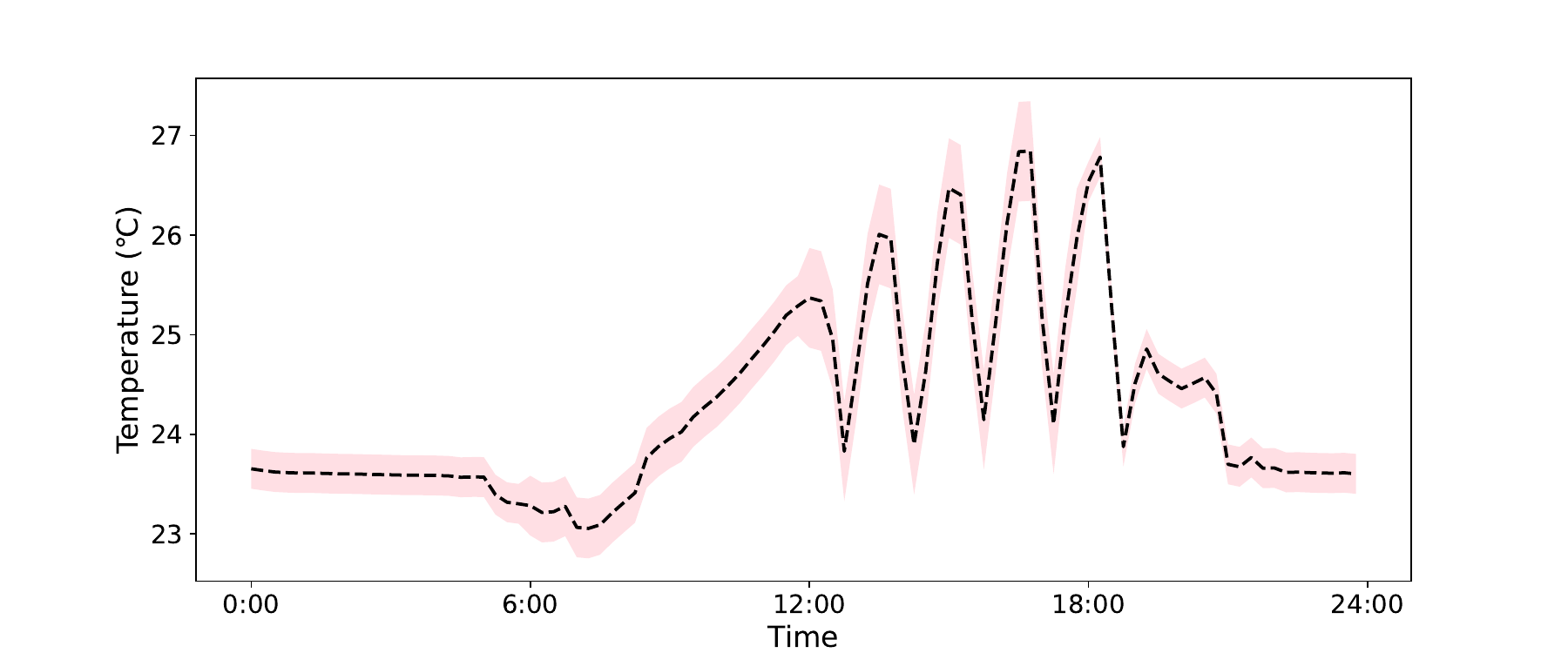}
\caption{\textcolor{black}{Comfort bounds on the 1st day of the training period.}}
\label{fig_comfort}
\end{figure}

\begin{table}[t]
\caption{Parameter Settings in Case Studies}
\centering
\resizebox{0.7\linewidth}{!}{\begin{tabular}{cc}
\hline
\multicolumn{1}{c|}{\textbf{Parameters}} & \textbf{Configurations} \\ \hline
\multicolumn{2}{c}{Variable formulation for 90-zone case}                       \\ \hline
\multicolumn{1}{c|}{AE of states}        & {[}64,32,16{]}          \\
\multicolumn{1}{c|}{AE of decisions}     & {[}64,32,16{]}          \\
\multicolumn{1}{c|}{AE of disturbances}  & {[}128,64,32{]}         \\
\multicolumn{1}{c|}{Activation function and epochs}        & ReLU, 2000          \\ \hline
\multicolumn{2}{c}{\textcolor{black}{Variable formulation for 1080-zone case}}                       \\ \hline
\multicolumn{1}{c|}{AE of states}        & {[}128,64,32{]}          \\
\multicolumn{1}{c|}{AE of decisions}     & {[}128,64,32{]}          \\
\multicolumn{1}{c|}{AE of disturbances}  & {[}128,128,64{]}         \\
\multicolumn{1}{c|}{Activation function and epochs}        & ReLU, 2000          \\ \hline
\multicolumn{2}{c}{Optimization}                          \\ \hline
\multicolumn{1}{c|}{Coefficient of performance}                 & 3.6                     \\
\multicolumn{1}{c|}{Power limit}       & 15 kW                    \\
\multicolumn{1}{c|}{Power penalty} & 10            \\
\multicolumn{1}{c|}{Temperature penalty} & 0.002 (\$/℃$\cdot$h)            \\
\multicolumn{1}{c|}{Price signal}        & NYISO 2023~\cite{case-price}              \\ \hline
\end{tabular}}
\label{tab-param}
\end{table}

\subsubsection{Comparisons}
To verify the effectiveness of the proposed method, we first solve ``Ground-truth" results based on the black-box building. The optimal solutions are obtained by numerical methods and the gradient descent algorithm. By setting an interval $\Delta {A}_{i,t}$, the numerical differentiation can be formulated as
\begin{equation}\label{numer}
 \nabla_{{a}_{i,t}} \overline C = \frac{\overline C({a}_{i,t}+\Delta {a}_{i,t})-\overline C({a}_{i,t})}{\Delta {a}_{i,t}},~\forall i,t
\end{equation}

We then set both identification and simulation models with original variables, denoted as \textbf{OriIden} and \textbf{OriSim}, respectively. \textcolor{black}{For OriIden, we identify a linear model and thus solve it with commercial solvers; we apply the zeroth-order optimization algorithm for OriSim. More specifically, OriIden and OriSim use original variables instead of the learned latent variables in Subsection IV.A for optimization. OriIden identifies a linear thermal dynamics model and incorporates it into the optimization problem of \eqref{ori_pro}; OriSim follows a similar way as OptSim (Fig.~\ref{fig_med_freeopt}) but solves with the original variables directly. Mathematically, OriIden differentiates from the Ground-truth and OptIden in the calculation of gradients for problem solving; OriSim differentiates from the Ground-truth in the iteration mechanism, and it differentiates from the OptSim in the update of variable types.} Detailed settings for all comparison methods are explained in Table~\ref{tab-compari}.
\begin{table}[t]
\caption{Comparison Settings}
\centering
\resizebox{0.95\linewidth}{!}{\begin{tabular}{cccc}
\hline
\textbf{Notations} & \textbf{Variable} & \textbf{Types} & \textbf{Optimization}                \\ \hline
Ground-truth    & Original   & / & Numerical differentiation                \\
OriIden           & Original  & Linear model & Commerical solver (Gurobi~\cite{gurobi})      \\
OriSim           & Original  & Simulation & Zeroth order optimization      \\
OptIden      & Latent & Linear model  & Auto-differentiation   \\
OptSim      & Latent & Simulation & Zeroth order optimization   \\ \hline
\end{tabular}}
\label{tab-compari}
\end{table}

\subsection{Latent Variables Representation}
Based on the proposed method, latent variables are represented from the collected data. As shown in Table~\ref{tab-dim}, the dimensions of the original variables are around 20-30 times more than the latent variables.
\begin{table}[t]
\caption{Dimensions of Original and Latent Variables}
\centering
\resizebox{0.85\linewidth}{!}{\begin{tabular}{r|cccc}
\hline
\textbf{Variables per time step} & \textbf{State} & \textbf{Decision} & \textbf{Disturbance} & \textbf{Total} \\ \hline
Original space     & 90                      & 80                           & 181                  & 351            \\
Latent space       & 3                       & 4                            & 6                    & 13             \\ \hline
\end{tabular}}
\label{tab-dim}
\end{table}

Furthermore, we compare the accuracy of identification models in OriIden and OptIden in Table~\ref{tab-acc}. \textcolor{black}{Three commonly used metrics are adopted, including root mean square error (RMSE), mean absolute error (MAE), and R-square (R2)~\cite{new-vpp}. The results are calculated by inputting the actual and predicted values of each zone's indoor temperature during the training and test periods, respectively. Both the mean and standard deviation values among 90 zones are presented in the table. Note that the errors for both OriIden and OptIden follow the same calculation method.} Compared with OriIden, the proposed approach can reduce RMSEs by around 29\% and 16\% in the training and test sets, respectively. Similar results are shown in other metrics, demonstrating a high potential for representing complex dynamics with low-dimensional variables.
\begin{table}[t]
\caption{Model Errors}
\centering
\resizebox{0.95\linewidth}{!}{\begin{tabular}{c|c|c|c|c}
\hline
\textbf{Methods}                         & \textbf{Dataset} & \textbf{RMSE (℃)} & \textbf{MAE (℃)} & \textbf{R2} \\ \hline
\multirow{2}{*}{OriIden}          & Train   & $0.4084\pm0.08$            & $0.3087\pm0.06$          & $0.8325\pm0.03$     \\
                                   & Test    & $0.4152\pm0.09$           & $0.3174\pm0.07$          & $0.8223\pm0.05$ \\ \hline             
\multirow{2}{*}{OptIden} & Train   & $0.2880\pm0.07$           & $0.1942\pm0.05$          & $0.9182\pm0.02$     \\
                                   & Test    & $0.3482\pm0.09$           & $0.2400\pm0.06$          & $0.8778\pm0.03$      \\ \hline
\end{tabular}}
\label{tab-acc}
\end{table}

\subsection{Optimization Results}
With latent variables, we conduct day-ahead optimization for all 31 days during the test period, and the solving results are analyzed in this subsection.
\subsubsection{Whole period results}
We first present actual costs for all comparisons in Fig.~\ref{fig_allday}. Compared with the ground truth, the results of all methods are subject to errors. The proposed methods have significantly lower gaps than OriIden and OriSim on most days. Roughly speaking, methods ordered by cost values are OriIden, OriSim, OptIden, and OptSim. This suggests that solving the optimization problem with latent variables and simulation models could realize lower costs than with original variables and identification models, respectively. Besides that, the proposed methods have more stable results in costs and the changing trend is consistent with the ground-truth during the whole period. On the contrary, the results from OriIden and OriSim change frequently and randomly, exhibiting unreliable optimization decisions for temperature control.
\begin{figure*}[t]
\centering
\includegraphics[width=0.99\textwidth]{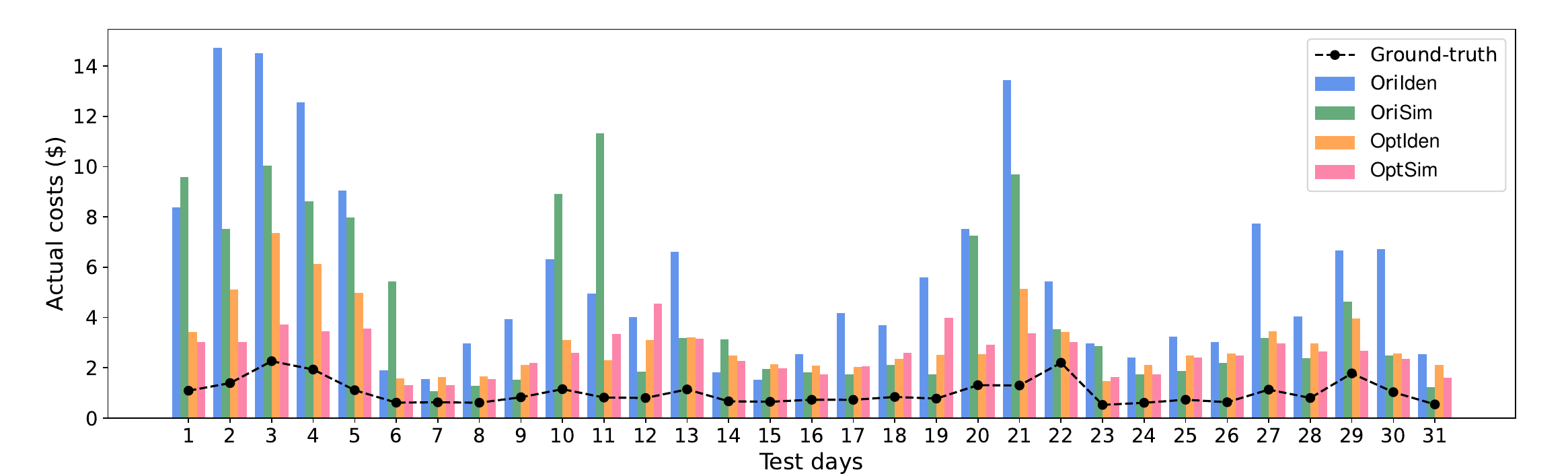}
\caption{Actual costs after optimization during the whole period.}
\label{fig_allday}
\end{figure*}

In Table~\ref{tab-avgcosts}, we further verify the above observations by presenting the mean and standard deviation values for the whole period. Compared with OriIden and OriSim, the proposed methods realize lower values in the mean and standard deviation of daily costs. Therefore, the proposed methods can obtain decisions that result in lower costs and are more consistent with ground-truth solutions.
\begin{table*}[t]
\caption{Whole-Period Statistical Results (\$)}
\centering
\resizebox{0.6\linewidth}{!}{\begin{tabular}{c|ccccc}
\hline
\textbf{Methods} & \textbf{Ground-truth} & \textbf{OriIden} & \textbf{OriSim} & \textbf{OptIden}& \textbf{OptSim}                \\ \hline
Mean  & 1.0113   & 5.6943 & 4.3182 & 3.0382 & 2.6201 \\ \hline
Standard deviation  & 0.4671   & 3.7406 & 3.1919 & 1.3709 & 0.8012 \\ \hline
\end{tabular}}
\label{tab-avgcosts}
\end{table*}

\subsubsection{Specific case analysis}
We specify the first day in the test period for detailed performance analysis. The solved (``\_dec") and actual (``\_act") results in Section V are shown in Table~\ref{tab-cost-1}. The results on actual costs are identical to Fig.~\ref{fig_allday}, where the proposed methods contribute to lower errors with ground truth compared with original variables. For errors between solved and actual results, the optimization based on simulation models with either original or latent variables could result in lower errors (around 0.2-0.4) than based on identification models due to the direct interaction with the real building. More specifically, for the two kinds of identification models, OptIden has a lower error (1.8141) than OriIden (7.8233) because it is more accurate in Table~\ref{tab-acc} to represent thermal dynamics.
\begin{table}[t]
\caption{Detailed Comparison of Costs (\$)}
\centering
\resizebox{0.95\linewidth}{!}{\begin{tabular}{c|c|c|c|c|c}
\hline
         & \textbf{Ground-truth} & \textbf{OriIden}& \textbf{OriSim} & \textbf{OptIden} & \textbf{OptSim} \\ \hline
Pow\_dec &   -           & 0.4125   &    5.2874  &    0.9731           & 1.0540       \\
Pow\_act & 0.3875       & 0.4125     & 4.2335   &     0.9741          & 1.0838      \\ \hline
Tem\_dec &    -          & 0.1497  & 4.6618   &    0.6448           & 2.2061       \\
Tem\_act & 0.7017       & 7.9730   & 5.3463  &  2.4580         & 1.9477       \\ \hline
Sum\_dec &     -         & 0.5622 & 9.9492   &     1.6179          & 3.2601       \\
\textbf{Sum\_act} & \textbf{1.0892}       & \textbf{8.3855}  &    \textbf{9.5798}   & \textbf{3.4320}           & \textbf{3.0315}      \\ \hline
\end{tabular}}
\label{tab-cost-1}
\end{table}

Then, we verify the detailed performance in each zone. We compare the temperature violation penalty of each method with the ground-truth result. The zones with large temperature violations are denoted as ``Abnormal" ones, while others are ``Normal" ones. \textcolor{black}{The detailed classification process is presented as:
\begin{subequations}\label{class}
\begin{align}
&\mathcal{U}_{i,x/0} = \mathbb{I}\left(|\frac{u_{i,x/0}-u_{i,0}}{u_{i,0}+\varsigma}|>U\right) \\
&u_{i,x} = \sum\limits_{t \in {\cal T}} P_{t} \cdot y_{i,x,t}^\intercal y_{i,x,t}
\end{align}
\end{subequations}
where $\mathcal{U}$ means the set of ``Normal/Abnormal" cases, with 0 as the ``Normal" and 1 as the ``Abnormal" case; $i\in\mathcal{Z}$ is the zone index and $x$ is the index of all comparisons $\{\text{Ground-truth, OriIden, OriSim, OptIden, OptSim}\}$ in the case study; specifically, the index $0$ is the method of ``Ground-truth"; $\mathbb{I}$ is the indicator function; $u$ is the penalty of temperature violation that is calculated by~\eqref{ori_1}; $\varsigma$ is a minimal value introduced to avoid division by zero; $U$ is the threshold to classify the ``Normal/Abnormal" cases, and its value is specified as 12, balancing the number of ``Abnormal" cases and the actual magnitude of temperature violation penalty.}

The results are presented in Fig.~\ref{fig_zones}, where the proposed methods have less than 10 zones with abnormal violations, but the original variable-related methods have more than 30 zones with abnormal violations. The results demonstrate that the proposed method can effectively regulate the temperature of most zones within comfort levels.
\begin{figure}[t]
\centering
\includegraphics[width=0.48\textwidth]{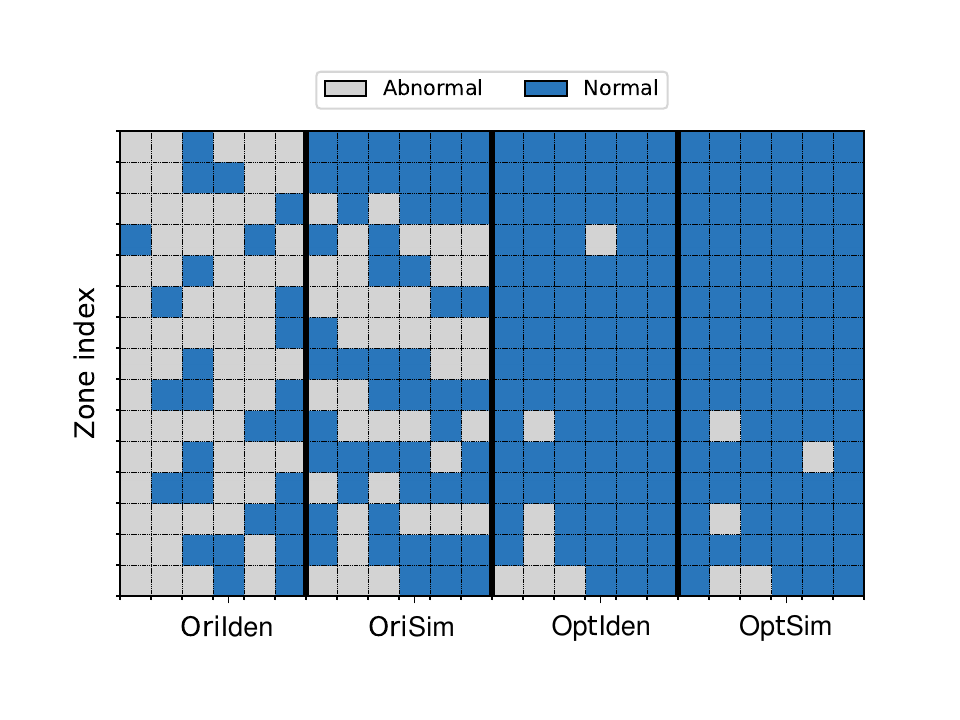}
\caption{Optimization results of comfort violations.}
\label{fig_zones}
\end{figure}

More specifically, we select the first zone in the building to analyze the daily temperature curves. The decision errors of each method between solved and actual results are presented in Fig.~\ref{fig_err1}; the actual temperatures solved by original variables and latent variables are presented in Figs.~\ref{fig_err2} and~\ref{fig_err3}, respectively, \textcolor{black}{where the shaded pink area represents the time-varying temperature comfort range with upper and lower bounds.} It can be found that both the OriSim and OptSim achieve almost zero errors in Fig.~\ref{fig_err1}. OptIden and OptSim in Fig.~\ref{fig_err3} realize smaller temperature violations than OriIden and OriSim in Fig.~\ref{fig_err2}, especially during 6:00-8:00 and 12:00-15:00 during the day.
\begin{figure}[t]
    \centering
    \subfloat[Temperature errors]{\includegraphics[width=0.48\textwidth,height=0.18\textwidth]{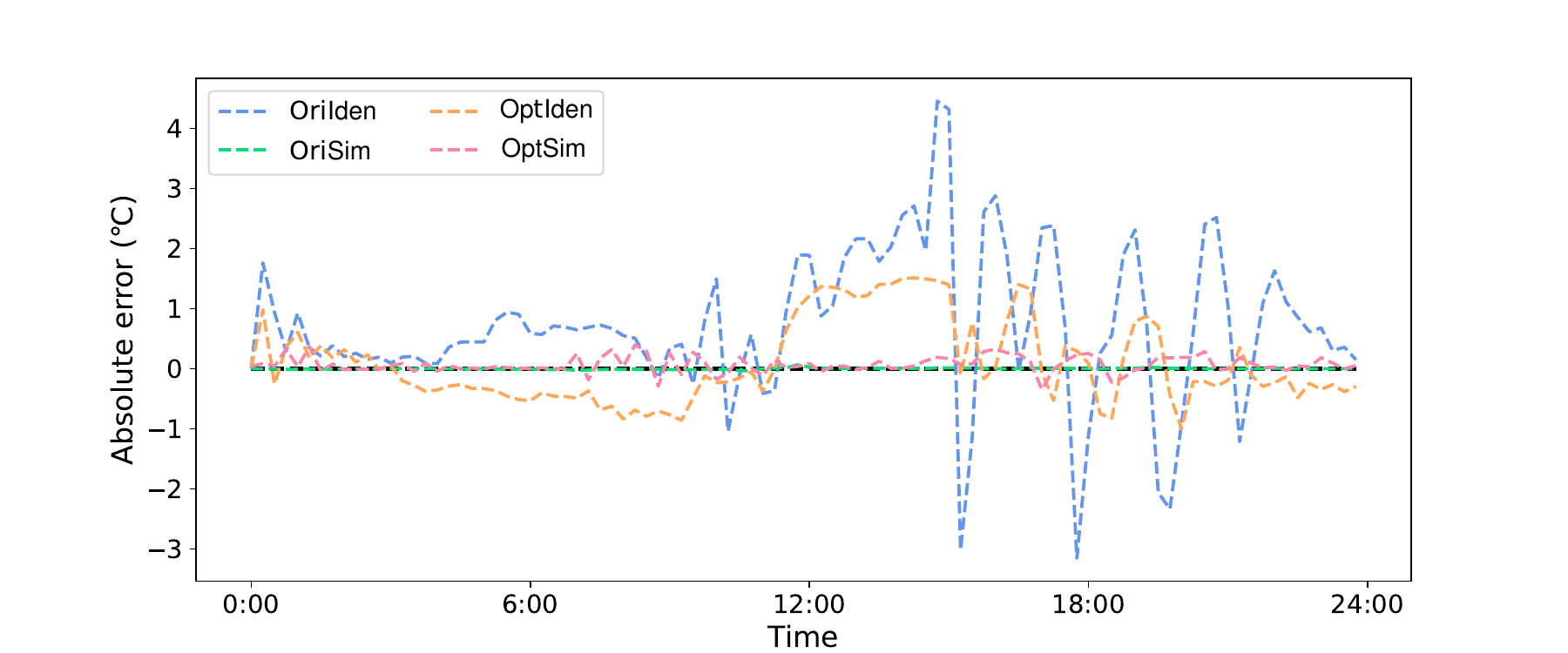}\label{fig_err1}} \\
    \subfloat[\textcolor{black}{Results solved by original variables}]{\includegraphics[width=0.48\textwidth,height=0.18\textwidth]{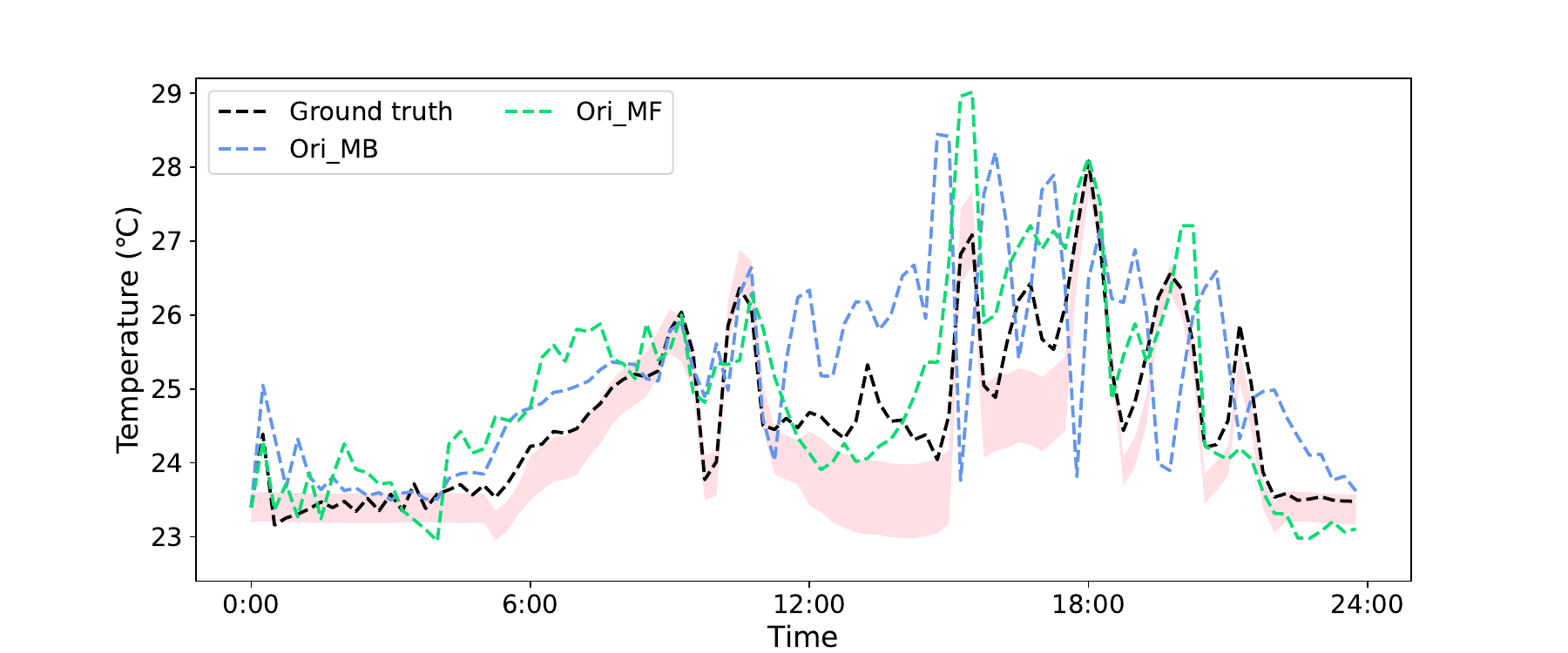}\label{fig_err2}} \\
    \subfloat[\textcolor{black}{Results solved by latent variables}]{\includegraphics[width=0.48\textwidth,height=0.18\textwidth]{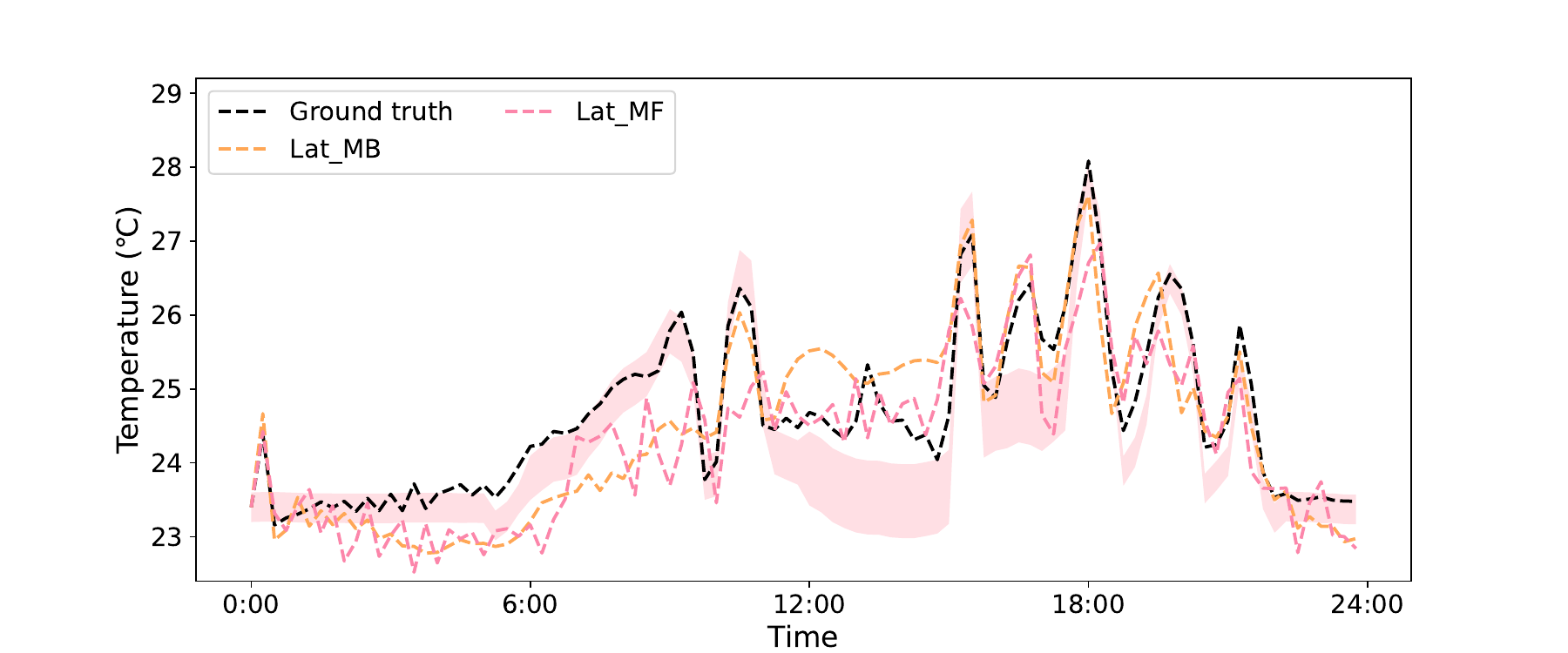}\label{fig_err3}}
    \caption{\textcolor{black}{Temperature curves in the 1st zone on the 1st day.}}
    \label{fig_errors}
\end{figure}

\subsection{Algorithm Performance}
We analyze some key algorithm performance in this subsection, including the computational time, algorithm convergence, and robustness to data noise. \textcolor{black}{Similar to the specific case analysis in Section VI.C, we select the first day in the test period for illustration.}
\subsubsection{Computational time}
Table~\ref{tab-time} shows the computation comparisons. The time of the ground-truth method is the highest due to numerical gradient calculations for each variable  ($|\mathcal{Z}| \cdot |\mathcal{T}|$). The proposed OptIden is around 6,000 times faster than the ground-truth method based on automatic differentiation. Note that the total computational time of OptSim (around 103s) is larger than OptIden (around 7s), mainly because OptSim requires more iterations for convergence without explicit gradient information.
\begin{table}[t]
\caption{Comparison of Computational Time (s)}
\centering
\resizebox{0.95\linewidth}{!}{\begin{tabular}{c|ccccc}
\hline
\textbf{Settings} & \textbf{Ground-truth} & \textbf{OriIden} & \textbf{OriSim} & \textbf{OptIden}& \textbf{OptSim}                \\ \hline
Total time    & 40299.978   & 355.506 & 211.031 & 7.327 & 103.816                \\
Per-iter time    & 822.449   & - & 0.107 & 0.129 & 0.208 \\ \hline
\end{tabular}}
\label{tab-time}
\end{table}

\subsubsection{Convergence} Considering the iterative solving process, we present the iteration losses of OriSim, OptIden, and OptSim in Fig.~\ref{fig_losses}. Benefiting from the calculated gradients, the interactive process of OptIden is the best, converging to less than 2.0 with around 60 iterations. Besides, in OptSim, the zeroth order optimization converges to around 4.0 with 500 iterations. On the contrary, with the same initial setting, the interactive process in OriSim can only reach a loss of 10.0-12.0 with up to 2000 iterations, and it has volatile loss changes during the early stage, resulting in poor convergence. It demonstrates that the zeroth order optimization with latent variables tends to achieve more effective convergence than with original variables.
\begin{figure}[t]
    \centering
    \subfloat[OriSim]{\includegraphics[width=0.48\textwidth,height=0.18\textwidth]{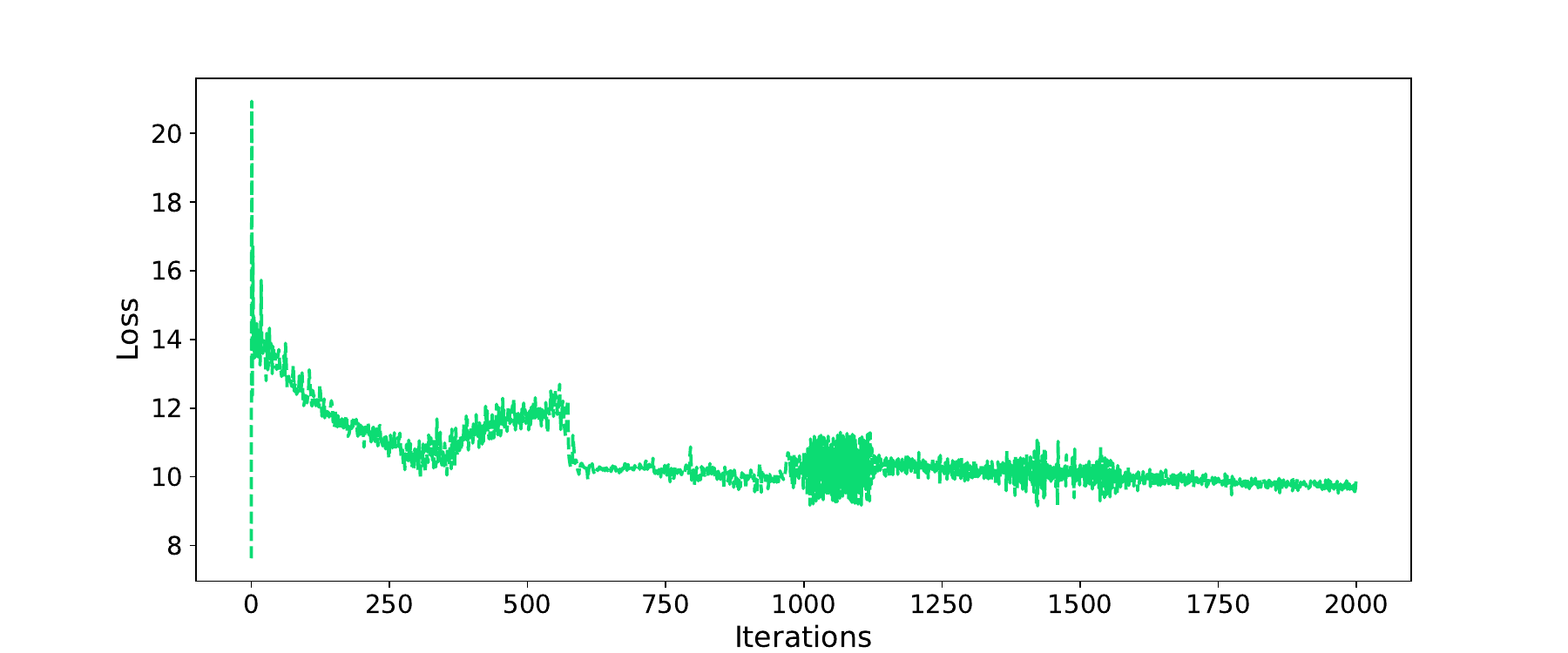}\label{fig_loss0}} \\
    \subfloat[OptIden]{\includegraphics[width=0.48\textwidth,height=0.18\textwidth]{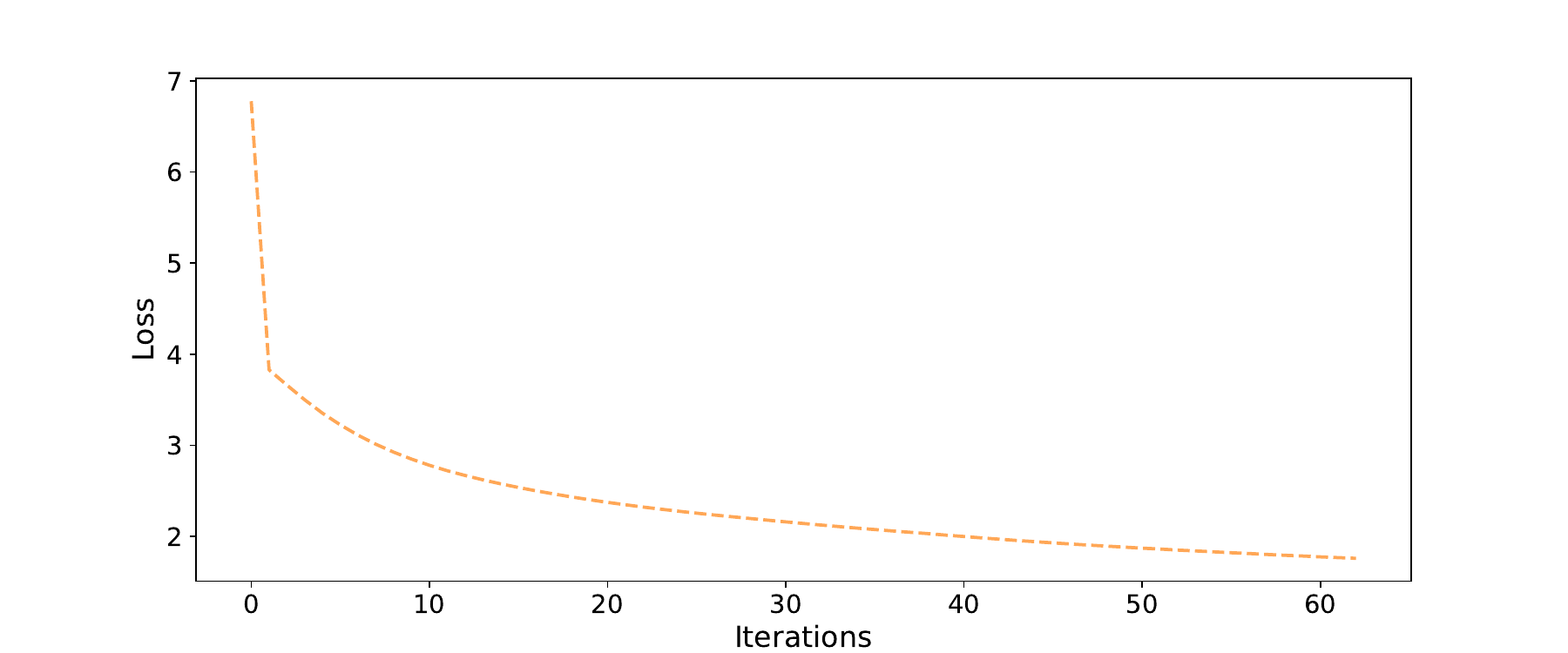}\label{fig_loss1}} \\
    \subfloat[OptSim]{\includegraphics[width=0.48\textwidth,height=0.18\textwidth]{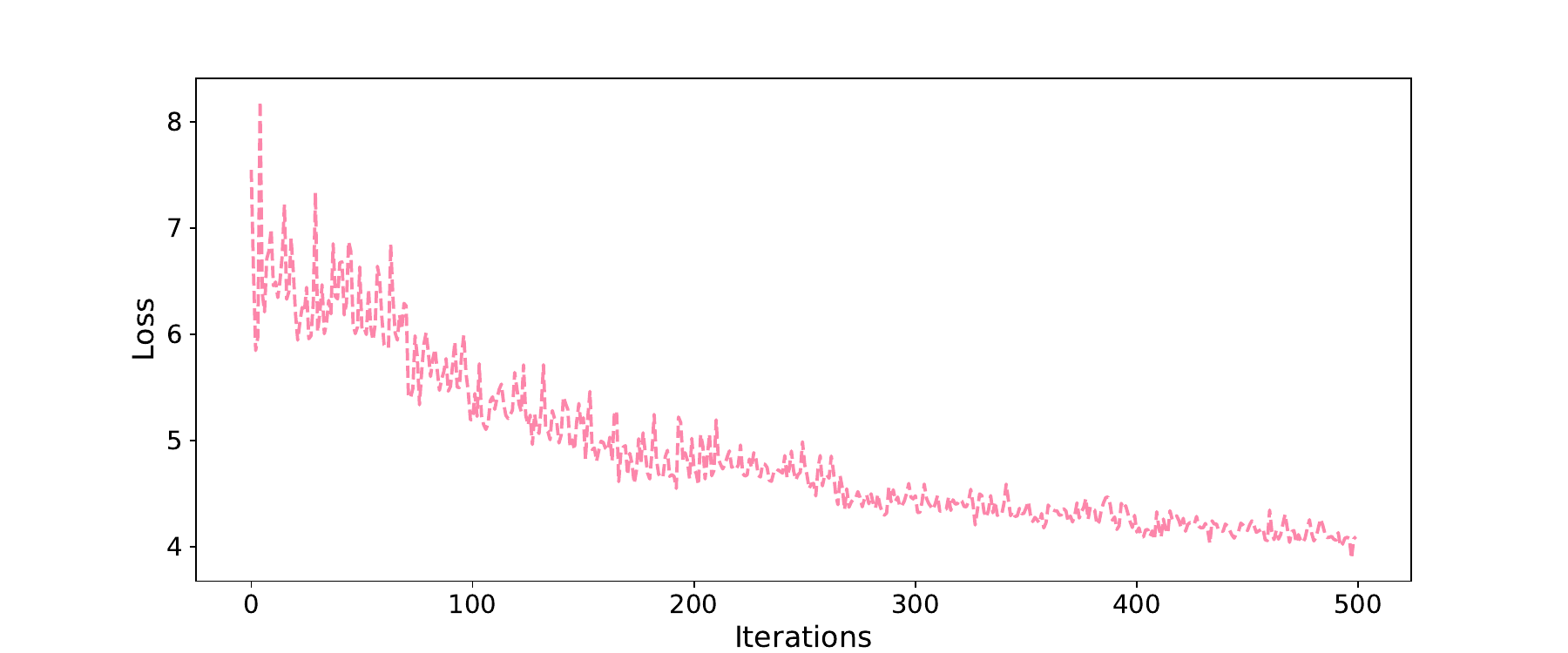}\label{fig_loss2}}
    \caption{Convergence process of different methods.}
    \label{fig_losses}
\end{figure}

\subsubsection{Robustness against noises}
Considering the underlying noises in the measured data for optimization, we add noise in the disturbance data to test the robustness of different methods. It is assumed that the added noises follow the normal distribution $\mathcal{N}$ as
\begin{equation}
\delta' = \delta\cdot(1+\varepsilon),~\varepsilon \sim \mathcal{N}(0,\sigma^2)
\end{equation}
where $\delta'$ and $\delta$ are the noisy and true data, respectively; $\varepsilon$ is the added noise ratio and $\sigma^2$ is the variance.

We set 10 different noise levels where $\sigma$ ranges from 0.1 to 1.0. The actual costs under different noise levels are given in Fig.~\ref{fig_noise}. \textcolor{black}{The calculation of actual costs is the same as the ``Sum\_act" in Table~\ref{tab-cost-1}.} The proposed methods are robust to noisy data because costs are less than 10\$, and most are less than 5\$ under different noise levels. The OriIden is sensitive to noise, and the cost increases to more than 10\$. The OriSim shows irregular changes with different noise levels, influenced by unstable iterations as shown in Fig.~\ref{fig_loss0}.
\begin{figure}[t]
\centering
\includegraphics[width=0.48\textwidth]{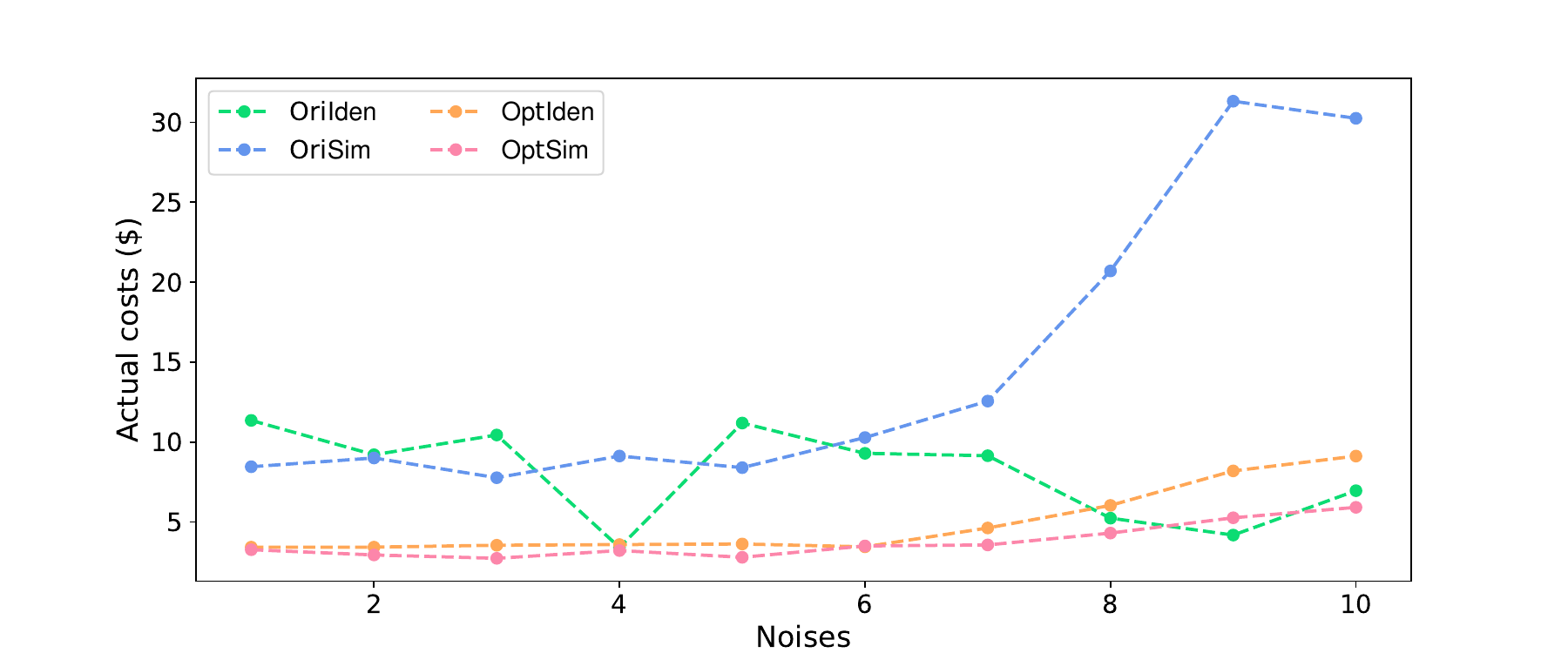}
\caption{Actual costs with noise settings.}
\label{fig_noise}
\end{figure}

\subsection{\textcolor{black}{Scalability Analysis}}

\textcolor{black}{We facilitate the 1080-zone building case to verify the scalability of the proposed method. As shown in Table~\ref{tab-dim-1080}, we represent latent state, action, and disturbance variables with dimensions around 1/300 of the original ones. Moreover, the modeling errors of OriIden and OptIden for the 1080-zone case are compared in Table~\ref{tab-acc-1080}. Compared with OriIden, the proposed approach can reduce RMSEs by around 55\% and 37\% in the training and test sets, respectively. Higher model accuracy of OptIden than OriIden can also be seen in MAE and R2, demonstrating the superiority of representing highly complex dynamics with a much lower dimension of variables.}
\begin{table}[t]
\caption{\textcolor{black}{Dimensions of Original and Latent Variables in the 1080-zone Case}}
\centering
\color{black}
\resizebox{0.85\linewidth}{!}{\begin{tabular}{r|cccc}
\hline
\textbf{Variables per time step} & \textbf{State} & \textbf{Decision} & \textbf{Disturbance} & \textbf{Total} \\ \hline
Original space     & 1080                      & 960                           & 2161                  & 4201            \\
Latent space       & 3                       & 5                            & 7                    & 15             \\ \hline
\end{tabular}}
\label{tab-dim-1080}
\end{table}
\begin{table}[t]
\caption{\textcolor{black}{Model Errors of 1080-zone case}}
\centering
\color{black}
\resizebox{0.95\linewidth}{!}{\begin{tabular}{c|c|c|c|c}
\hline
\textbf{Methods}                         & \textbf{Dataset} & \textbf{RMSE (℃)} & \textbf{MAE (℃)} & \textbf{R2} \\ \hline
\multirow{2}{*}{OriIden}          & Train   & $0.5274\pm0.11$            & $0.3968\pm0.08$          & $0.7413\pm0.04$     \\
                                   & Test    & $0.5605\pm0.11$           & $0.4333\pm0.09$          & $0.6932\pm0.08$ \\ \hline             
\multirow{2}{*}{OptIden} & Train   & $0.2356\pm0.06$           & $0.1565\pm0.04$          & $0.9478\pm0.01$     \\
                                   & Test    & $0.3520\pm0.08$           & $0.2368\pm0.05$          & $0.8813\pm0.03$      \\ \hline
\end{tabular}}
\label{tab-acc-1080}
\end{table}

\textcolor{black}{After representing latent variables, we present day-ahead optimization results for the test period in Fig.~\ref{fig_allday_1080}. Compared with OriIden and OriSim, which utilize the original variables for solving, the proposed OptIden and OptSim have significantly lower costs on most days. More specifically, OriIden and OriSim achieve an average cost of 101.29 and 48.67 during the test period, while the results of OptIden and OptSim are 54.39 and 24.38, respectively. Thus, it confirms that solving the optimization problem with latent variables and simulation models could realize lower costs than with original variables and identification models respectively, even in the much larger building case with 1080 zones.}
\begin{figure*}[t]
\centering
\includegraphics[width=0.99\textwidth]{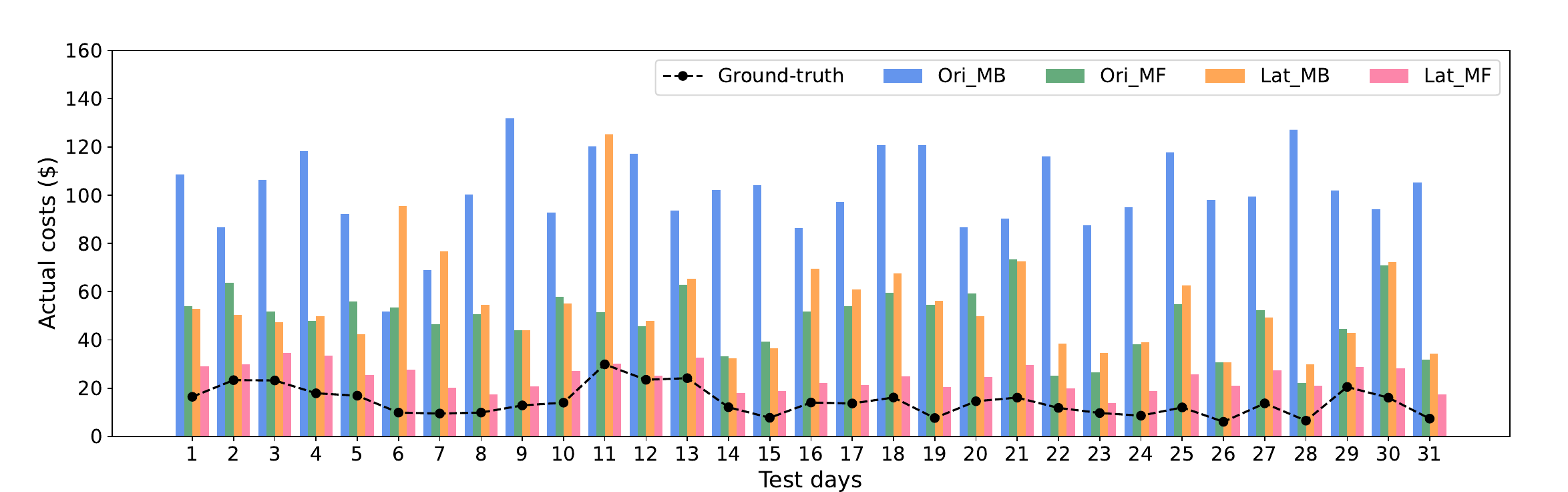}
\caption{\textcolor{black}{Comparison of actual costs in the 1080-zone case.}}
\label{fig_allday_1080}
\end{figure*}

\section{Conclusion}
This work solves a dimension-reduced optimization of multi-zone TCLs with latent variables. A multi-task learning framework is developed to represent latent variables and the time-series model jointly. Latent variables can be incorporated with identification and simulation models to reduce the problem's complexity. The proposed methods leveraging latent variables significantly improve computational efficiency by substantially reducing the number of variables in the optimization process. Additionally, they enhance optimization accuracy by using an auto-encoder to capture the interdependencies between system states effectively. For instance, the computational time of OptIden is approximately 6,000 times faster than the ground-truth method, while the average cost of OptSim is nearly 39\% lower than OriSim.

\textcolor{black}{When utilizing latent variables, the OptSim algorithm could achieve nearly 14\% lower costs by avoiding errors of the identification model, but it is solved with more iterations and adds around 14 times the computational time. As for practical applications, the OptIden is applicable to most buildings only with historical data as inputs, while the OptSim performs well when pursuing accurate control and if the interaction with real buildings is available.}

Several possible research are listed to extend the proposed method: \textcolor{black}{1) Incorporate the representation and modeling process of latent variables within optimization to reduce solving errors in an end-to-end way. 2) Coordinate the proposed OptSim with RL in the complex building environments. For example, the representation of low-dimensional latent variables could serve as compact inputs for RL to improve sampling efficiency in large-scale building systems~\cite{new-rl1}; the zeroth order optimization provides a computationally efficient approach for policy training in RL that could reduce the large number of interactions with the simulation environment~\cite{new-rl2}.}


\bibliographystyle{IEEEtran}
\bibliography{Re1_Ref}
\end{document}